\documentclass[a4paper,USenglish,cleveref,autoref,thm-restate,numberwithinsect]{lipics-v2021}

\hideLIPIcs
\nolinenumbers
\usepackage[utf8]{inputenc}
\usepackage{tikz}
\tikzstyle{vertex}=[draw, circle, fill, inner sep = 2.4pt]

\usepackage{todonotes}

\usepackage{xspace}

\newcommand{\F}{\ensuremath{\mathbb F}\xspace}

\newcommand{\cI}{\ensuremath{\mathcal I}\xspace}

\DeclareMathOperator{\Pf}{Pf}

\DeclareMathOperator{\poly}{poly}
\DeclareMathOperator{\supp}{supp}
\DeclareMathOperator{\osupp}{osupp}

\newcommand{\mcal}{\mathcal}
\newcommand{\monomial}{\mathbf{m}}
\newcommand{\matroid}{\mathcal{M}}
\newcommand{\watroid}{\mathcal{W}}
\newcommand{\sub}{\subseteq}
\newcommand{\set}[1]{\left\{#1\right\}}
\newcommand{\grp}[1]{\left(#1\right)}
\newcommand{\vertexcolor}{\textsc{Vertex Coloring}}
\newcommand{\edgecolor}{\textup{\textsc{{Edge Coloring}}}}
\newcommand{\listedgecolor}{\textup{\textsc{List Edge Coloring}}}
\newcommand{\coloring}{\textsc{Coloring}}
\newcommand{\NP}{\textsf{NP}}
\newcommand{\eps}{\varepsilon}
\newcommand{\FF}{\mathbb{F}}
    
    \let\emptyset\varnothing
\newcommand{\notD}{V'}
\newcommand{\datroid}{\mathcal{D}}
\newcommand{\new}{\textup{new}}
\newcommand{\Gnew}{G_{\new}}
\newcommand{\nnew}{n_{\new}}
\newcommand{\mnew}{m_{\new}}
\newcommand{\Vnew}{V_{\new}}
\newcommand{\Enew}{E_{\new}}
\newcommand{\matextend}{\matroid^*_{\downarrow}}

\definecolor{amethyst}{rgb}{0.6, 0.4, 0.8}

\title{Faster Edge Coloring by Partition Sieving}

\author{Shyan Akmal}{INSAIT, Sofia University ``St. Kliment Ohridski'', Bulgaria \and \url{https://www.shyan.akmal.com}}{shyan.akmal@insait.ai}{https://orcid.org/0000-0002-7266-2041}{Partially funded by the Ministry of Education and Science of Bulgaria (support
for INSAIT, part of the Bulgarian National Roadmap for Research Infrastructure).}

\author{Tomohiro Koana}{Technische Universität Berlin, Faculty~IV, Institute of Software Engineering and Theoretical Computer Science, Algorithmics and Computational Complexity}{tomohiro.koana@tu-berlin.de}{0000-0002-8684-0611}{Supported by the DFG Project DiPa, NI 369/21.}

\authorrunning{S. Akmal and T. Koana}


\acknowledgements{We thank the anonymous reviewers for helpful feedback on this work.}

\Copyright{Shyan Akmal and Tomohiro Koana}


\EventEditors{}
\EventNoEds{4}
\EventLongTitle{}
\EventShortTitle{}
\EventAcronym{}
\EventYear{}
\EventDate{}
\EventLocation{}
\EventLogo{}
\SeriesVolume{}
\ArticleNo{}

\begin{document}

\maketitle

\begin{abstract}

In the \edgecolor{} problem,
we are given an undirected graph $G$ with $n$ vertices and $m$ edges,
and are tasked with finding the smallest positive integer $k$ so that the edges of $G$ can be assigned $k$ colors in such a way that no two edges incident to the same vertex are assigned the same color. 
\edgecolor{} is a classic \NP-hard problem, and so significant research has gone into designing fast \emph{exponential-time} algorithms for solving \edgecolor{} and its variants exactly. 
Prior work showed that \edgecolor{} can be solved in $2^m\poly(n)$ time and polynomial space, and
in graphs with average degree $d$ in $2^{(1-\eps_d)m}\poly(n)$ time and exponential space, where $\eps_d = (1/d)^{\Theta(d^3)}$.

We present an algorithm that solves \edgecolor{} in $2^{m-3n/5}\poly(n)$ time and polynomial space. 
Our result is the first algorithm for this problem which simultaneously runs in faster than $2^m\poly(m)$ time and uses only polynomial space.
In graphs of average degree $d$, our algorithm runs in $2^{(1-6/(5d))m}\poly(n)$ time, which has far better dependence in $d$ than previous results. 
We also consider a generalization of \edgecolor{} called
  \textsc{List Edge Coloring},
  where each edge $e$ in the input graph 
  comes with a list $L_e\sub\set{1, \dots, k}$ of colors,
  and we must determine whether we can assign each edge a color from its list so that no two edges incident to the same vertex receive the same 
  color.
  We show that this problem can be solved in $2^{(1-6/(5k))m}\poly(n)$ time and polynomial space.
  The previous best algorithm for \listedgecolor{} took $2^m\poly(n)$ time and space. 

  Our algorithms are algebraic, and work by constructing a special polynomial $P$ based off the input graph that contains a multilinear monomial (i.e., a monomial where every variable has degree at most one) if and only if the answer to the \listedgecolor{} problem on the input graph is YES.
  We then solve the problem by detecting multilinear monomials in $P$. 
  Previous work also employed such monomial detection techniques to solve \edgecolor{}.
  We obtain faster algorithms 
  both by carefully constructing our polynomial $P$, and by improving the runtimes for certain structured monomial detection problems using a technique we call partition sieving.
\end{abstract}

\newpage 

\section{Introduction}

Coloring graphs is a rich area of research in graph theory and computer science. 
Given a graph $G$, a \emph{proper vertex coloring} of $G$ is an assignment of colors to its nodes such that no two adjacent vertices receive the same color. 
In the \vertexcolor{} problem, we are given an undirected graph $G$ on $n$ vertices, and are tasked with computing the smallest positive integer $k$ such that $G$ admits a proper vertex coloring using $k$ colors. 
\vertexcolor{} is a classic combinatorial problem, 
\NP-hard to solve even approximately \cite{LundYannakis1994}.
An influential line of research has worked on designing faster and faster exponential-time algorithms for this problem \cite{Lawler76,Eppstein2001,Byskov2004,BjorkHusf2007}, culminating in an algorithm solving \vertexcolor{} in $O^*(2^n)$ time and space \cite{BjHuKo09},
where throughout we write $O^*(f(n))$ as shorthand for $f(n)\poly(n)$.

In this paper, we study exact algorithms for the closely related 
\edgecolor{} problem.
Given a graph $G$, a \emph{proper edge coloring} of $G$ is an assignment of colors to its edges such that no two edges incident to a common vertex receive the same color.
The smallest positive integer $k$ such that $G$ admits a proper edge coloring using $k$ colors is the \emph{chromatic index} of~$G$, denoted by $\chi'(G)$.
In the \edgecolor{} problem, we are given an undirected graph $G$ on $n$ vertices and $m$ edges, and are tasked with computing $\chi'(G)$.

Let $\Delta$ denote the maximum degree of the input graph $G$. 
For each vertex $v$ in $G$, a proper edge coloring of $G$ must assign distinct colors to the edges incident to $v$.
Thus $\chi'(G)\ge \Delta$.
A classic result in graph theory known as
Vizing's theorem proves that in fact $\chi'(G)\le \Delta+1$,
so that the simple $\Delta$ lower bound is close to the truth. 
Moreover, a proper edge coloring of $G$ using $(\Delta+1)$ colors can be found in near-linear time~\cite{near-linear-vizing}. 
Consequently, solving \edgecolor{} amounts to distinguishing between the cases $\chi'(G) = \Delta$ and $\chi'(G)=\Delta+1$.
Despite this powerful structural result, the \edgecolor{} problem is \NP-hard just like \vertexcolor{}, even for graphs where all vertices have degree exactly $\Delta = 3$ \cite{Holyer81a}.

In recent years, researchers have made major strides in our knowledge of algorithms for \emph{approximate} variants of \edgecolor{} in the distributed \cite{Harris2019,Bernshteyn2022,BaBrKuOl2022} and dynamic settings \cite{Duan2019,chrRoVlie2024,BhCoMaPaSo2024}.
In comparison,
much less is known about the exponential-time complexity for solving the \edgecolor{} problem \emph{exactly}.

A simple way to solve \edgecolor{} is by reduction to \vertexcolor{}.
Given the input graph $G$,
we can construct in polynomial time the \emph{line graph} $L(G)$
whose nodes are edges of $G$, and whose nodes are adjacent if the corresponding edges in $G$ are incident to a common vertex.
By construction, the solution to \vertexcolor{} on $L(G)$ gives the solution to \edgecolor{} on $G$.
If $G$ has $m$ edges and maximum degree $\Delta$,
then $L(G)$ has $m$ vertices and maximum degree $(2\Delta-1)$.
Using the aforementioned $O^*(2^n)$ time and space algorithm for \vertexcolor{},
this immediately implies a $O^*(2^m)$ time and space algorithm for \edgecolor{}. 
In graphs with maximum degree $\Delta$,
\vertexcolor{} can be solved in $O^*(2^{(1-\eps)n})$ time and space,
where $\eps = (1/2)^{\Theta(\Delta)}$ \cite{BjorklundHusfeldtKaskiKoivisto2009}.
By applying the above reduction, \edgecolor{} can similarly be solved in  $O^*(2^{(1-\eps)m})$ time and space, for $\eps = (1/2)^{\Theta(\Delta)}$.

Line graphs are much more structured objects than generic undirected graphs (for example, there are small patterns that line graphs cannot contain as induced subgraphs \cite[Theorem 3]{beineke1968}), so one might hope to solve \edgecolor{} faster without relying on a black-box reduction to \vertexcolor{}.
Nonetheless, there is only one result from previous work which solves \edgecolor{} on general undirected graphs without using this reduction.
Specifically, \cite[Section 5.6]{BjorklundHKK17narrow} shows that \edgecolor{} can be solved in $O^*(2^m)$ time and polynomial space.
In comparison, it remains open whether \vertexcolor{} can be solved in $O^*(2^n)$ time using only polynomial space. 

\noindent In summary, algorithms from previous work
can solve \edgecolor{} in 
\begin{itemize}
    \item  $O^*(2^m)$ time using polynomial space, or in
    \item faster than $O^*(2^m)$ time in bounded degree graphs, using exponential space. 
\end{itemize}

\noindent Given this state of affairs, it is natural to ask: 

\begin{center}
    \textit{Can \edgecolor{} be solved in faster than $O^*(2^m)$ time, using only polynomial space? }
\end{center}

\noindent We answer this question affirmatively by proving the following result. 

\begin{restatable}{theorem}{main}
\label{thm:edgecolor}
There is a randomized algorithm which solves \edgecolor{} with high probability and one-sided error in $O^*(2^{m - 3n/5})$ time and polynomial space.
\end{restatable}

If the input graph has average degree $d$,
then $m = dn/2$, so
\Cref{thm:edgecolor} 
 equivalently states that \edgecolor{} can be solved in 
$O^*(2^{(1-\eps)m})$ time for $\eps = \frac 6{5d}$.
In comparison, the current fastest algorithm for \vertexcolor{} on graphs with average degree $d$ takes $O^*(2^{(1-\delta)n})$ time, where $\delta = (1/d)^{\Theta(d^3)}$
\cite[Lemma 4.4 and Theorem 5.4]{DBLP:journals/talg/GolovnevKM16}.
Prior to our work, the fastest algorithm for \edgecolor{} in the special case of regular graphs (i.e., graphs where all vertices have the same degree)  took $O^*(2^{m-n/2})$ time \cite[Theorem 6]{BjorklundHKK17narrow}. 
The regularity assumption simplifies the \edgecolor{} problem significantly, as in this case the problem reduces to finding a collection of $(\Delta - 1)$ mutually disjoint perfect matchings, which together consist of only $(m - n/2)$ edges. 
\Cref{thm:edgecolor} improves upon this runtime with an algorithm that succeeds on \emph{all} undirected graphs, not just regular graphs. 

In the case of regular graphs, we obtain improvements beyond \Cref{thm:edgecolor} when the graphs have high degree.
We say a graph is $d$-regular if all its nodes have degree $d$. 
Given a positive integer $k$, let 
    \begin{equation}
    \label{eq:harmonic}
    H_k = \sum_{j=1}^k \frac 1j
    \end{equation}
denote the $k^{\text{th}}$ Harmonic number. 
We prove the following result. 

\begin{restatable}{theorem}{regular}
\label{thm:regular}
For each integer $d\ge 6$, there is a randomized algorithm that solves \edgecolor{} with high probability and one-sided error on $d$-regular graphs in $O^*(2^{m - \alpha_d n})$ time and polynomial space, for $\alpha_d = 1 - H_{d+1}/(d+1)$.
\end{restatable}

For example for $d=6$, \Cref{thm:regular} shows that we can solve \edgecolor{} in 6-regular graphs in $O^*(2^{m - 0.62n})$ time, slightly faster than the runtime presented in \Cref{thm:edgecolor}.
The savings in the exponent are better for larger $d$, approaching a runtime of $O^*(2^{m-n})$ as the degree $d$ grows
since we have $\alpha_d = 1 - \Theta((\log d)/d)$. 

We also consider a generalization of \edgecolor{} called \listedgecolor{}.
In this problem, we are given 
an undirected graph $G$ with $m$ edges as before, together with a positive integer $k$ and a list of possible colors $L_e\sub \set{1, \dots, k}$
for each edge $e$. 
We are tasked with determining whether $G$ admits a proper edge coloring, which assigns each edge $e$ a color from the list $L_e$.
If we set $k = \Delta$ to be the maximum degree of $G$ and set $L_e = \set{1, \dots, k}$ for each edge $e$, we recover the standard \edgecolor{} problem.
Using more sophisticated arguments, we generalize \Cref{thm:edgecolor} to  the \listedgecolor{} problem.

\begin{restatable}{theorem}{mainwlist}
\label{thm:main}
There is a randomized algorithm that solves {\normalfont\textsc{List Edge Coloring}} with high probability and one-sided error in $O^*(2^{m - 3n/5})$ time and polynomial space.
\end{restatable}

Prior to our work, no algorithm was known for \listedgecolor{} which simultaneously ran in faster than $O^*(2^m)$ time and used polynomial space. 

\subsection{Our Techniques}

Our algorithms for \edgecolor{} and \listedgecolor{} are algebraic,
and involve reducing these problems to certain \emph{monomial detection} tasks,
by now a common paradigm in graph algorithms.
To achieve the $O^*(2^{m-3n/5})$ runtime and polynomial space bounds for these problems, we combine  and improve techniques in polynomial sieving, matroid constructions, and enumeration using matrices. 

The main technical ingredient in our work builds off the \emph{determinantal sieving} technique introduced in \cite{determinantalsieving}.
Given a homogeneous multivariate polynomial $P(X)$ of degree $k$ and a linear matroid $\matroid$ on $X$ of rank $k$, this technique allows one to test  in $O^*(2^k)$ time whether $P(X)$ contains a multilinear monomial (i.e., a monomial where each variable has degree at most one) that forms a basis in $\matroid$.
We recall the formal definition of matroids in \Cref{sec:prelim}.
For the purpose of this overview, a matroid $\matroid$ is a family of subsets of variables of $P$ satisfying certain properties,
and a basis in $\matroid$ is a set in this family whose size equals the rank $k$.
In our applications, we construct $P$ to enumerate certain structures in the input graph.
Finding a multilinear monomial in $P$ corresponds to identifying a particularly nice structure in the graph (e.g., a solution to the \edgecolor{} problem). 
Determinantal sieving is useful because one can pick $\matroid$ in such a way that identifying a monomial in $P$ that forms a basis in $\matroid$ recovers extra information about the relevant structures in the graph.

We build off determinantal sieving and introduce the \emph{partition sieving} method.
This technique involves a \emph{partition matroid} $\datroid$, which is a simple type of matroid defined using an underlying partition of the variable set of $P$. 
Given a partition matroid $\datroid$ whose partition has $p$ parts, and a polynomial $P$ \emph{compatible} with~$\datroid$ 
(``compatible'' is a technical condition defined in \Cref{sec:polynomial-sieving}---intuitively, it requires that monomials in $P$ have degrees respecting the partition defining $\datroid$), we can  test whether $P$ contains a multilinear monomial forming a basis in $\datroid$ in $O^*(2^{k - p})$ time and polynomial space.
This offers an exponential speed-up over determinantal sieving.
To apply partition sieving, we need to carefully design both the polynomial $P$ and the partition matroid $\datroid$ to be compatible in our technical framework.

We design the polynomial $P$ using the notion of the Pfaffian of a matrix. 
If $A$ is a symbolic skew-symmetric adjacency matrix of $G$,
then its Pfaffian $\Pf A$ is a polynomial that enumerates the perfect matchings in $G$. 
The \edgecolor{} problem may be rephrased as asking whether the edge set of $G$ can be partitioned into a collection of $k$ matchings.
One can design a polynomial that enumerates $k$ matchings by computing a product of $k$ Pfaffians of adjacency matrices.
However, this simple construction does not meet the technical conditions we need to employ partition sieving. 
Instead we utilize the Ishikawa-Wakayama formula \cite{IshikawaW95}, 
a convolution identity for Pfaffians and determinants, 
to design a polynomial that enumerates  matchings which  satisfy certain degree constraints.
Specifically, the formula lets us construct a polynomial $P$ whose monomials correspond to $k$-tuples of matchings in $G$, where each vertex in $G$ appears as an endpoint in exactly $\deg_G(v)$ of the matchings in the tuple.

We design the partition matroid $\datroid$ in terms of a \emph{dominating set} of the input graph.
A subset of vertices $D$ is a dominating set in $G$ if every vertex in $V\setminus D$ is adjacent to a node in $D$. 
Given a dominating set $D$ in $G$,
we show that it is always possible to construct a partition matroid $\datroid$ with $p = |V\setminus D|$ parts, 
compatible with our enumerating polynomial $P$ (in our proofs, $\datroid$ is actually only compatible with a polynomial obtained from $P$ by the inclusion-exclusion principle, but we ignore that distinction in this overview).
We achieve compatability using the degree condition imposed on $P$, mentioned at the end of the previous paragraph. 
Partition sieving then lets us detect a multilinear monomial in $P$, and thus 
solve \edgecolor{}, in $O^*(2^{m-|V\setminus D|})$ time.

\begin{restatable}{lemma}{domset}
\label{lem:domset}
There is a randomized algorithm that,
given a graph $G$ on $n$ vertices and $m$ edges and a dominating set $D$ of $G$,
solves \edgecolor{} on $G$ with high probability and one-sided error in $O^*(2^{m - n +|D|})$ time and polynomial space.
\end{restatable}

In other words, by finding smaller dominating sets in $G$, we can solve \edgecolor{} faster. 
In polynomial time, it is easy to construct a dominating set of size at most $n/2$ in any connected graph $G$ with $n$ vertices.
Consequently, the partition sieving framework lets us solve \edgecolor{} in $O^*(2^{m-n/2})$ time.
We obtain faster algorithms using the fact that if a graph with $n\ge 8$ nodes has minimum degree two, then it contains a dominating set of size at most $2n/5$ \cite{mccuaig1989domination}. 
We make this result effective, showing how to \emph{find} such a dominating set in $O^*(2^{m-3n/5})$ time and polynomial space. 
When solving \edgecolor{}, it turns out one can assume without loss of generality that the graph has minimum degree two, and so this framework yields a $O^*(2^{m-3n/5})$ time algorithm for the problem. 
The situation is far more complex for \textsc{List Edge Coloring},
where unit-degree vertices cannot be removed freely.
For this problem, we construct a more complicated polynomial $P$, by enforcing additional matroid conditions in the Ishikawa-Wakayama formula.
Intuitively, we identify a subgraph $\Gnew$ of minimum degree two in $G$, and use $P$ to determine whether $\Gnew$ admits a list edge coloring which can be \emph{extended} to all of $G$.

\paragraph*{Organization}
In 
\Cref{sec:prelim} we review 
notation, basic assumptions, and
useful facts about polynomials, matrices, and matroids.
In \Cref{sec:polynomial-sieving} we introduce our partition sieving method for monomial detection. 
In \Cref{sec:template} we present a generic algorithmic template for solving coloring problems using polynomials. 
In \Cref{sec:edge-color} we apply this framework to design our algorithms for \edgecolor{} and prove \Cref{thm:edgecolor,thm:regular}, and 
in \Cref{sec:list-color} we apply this framework to design our algorithm for \listedgecolor{} and prove \Cref{thm:main}.
Although \Cref{thm:main} subsumes \Cref{thm:edgecolor}, we include proofs of both results separately, since the proof of \Cref{thm:main} is significantly more complicated. 
We conclude in \Cref{sec:conclusion} with a summary of our work and some open problems. 
See \Cref{sec:related} for a discussion of additional related work, 
\Cref{sec:construction} for omitted proofs concerning matroid and Pfaffian constructions, and 
\Cref{sec:dom} for our dominating set algorithms.

\section{Preliminaries} 
    \label{sec:prelim}

\paragraph*{General Notation}

Given a positive integer $a$, we let $[a] = \{1, \dotsc, a\}$
denote the set of the first $a$ consecutive positive integers. 
Given a set $S$ and positive integer $k$, we let $\binom{S}{k}$ denote the family of subsets of $S$ of size $k$. 
Given a positive integer $k$,
we let $H_k$ denote the $k^{\text{th}}$ Harmonic number, defined in \Cref{eq:harmonic}.

\paragraph*{Graph Notation and Assumptions}

Throughout, we consider graphs which are undirected.
We let  $n$ and $m$ denote the number of vertices and edges in the input graphs for the problems we study. 
We assume that the input graphs are connected, so that $m\ge n-1$.
This is without loss of generality, since when solving \edgecolor{} and \listedgecolor{} on disconnected graphs, it suffices to solve the problems separately on each connected component. 
Given a graph $G$ and vertex $v$, we let $\deg_G(v)$ denote the number of neighbors of $v$ (i.e., the degree of $v$ in $G$). 
For the \listedgecolor{} problem where each edge is assigned a list of colors from a subset of $[k]$,
we assume that $\deg_G(v)\le k$ for all vertices $v$ in $G$, since if some vertex $v$ has degree greater than $k$, its incident edges must all be assigned distinct colors in a proper edge coloring, and the answer to the \listedgecolor{} problem is trivially no. 

We let $\Pi(G)$ denote the set of perfect matchings of $G$---partitions of the vertex set of $G$ into parts of size two, such that each part consists of two adjacent vertices. 

\paragraph*{Dominating Sets}

A dominating set $D$ in a graph $G$ is a subset of vertices of $G$ such that every vertex in $G$ is either in $D$ or adjacent to a node in $D$. 
We collect some useful results relating to finding small dominating sets in graphs. 

\begin{restatable}[Simple Dominating Set]{lemma}{simpledomset}
    \label{lem:basic-dom}
    There is a polynomial time algorithm which takes in
    a connected graph on $n$ nodes and outputs a dominating set for the graph of size at most $n/2$.
\end{restatable}

\Cref{lem:basic-dom} is due to Ore~\cite{ore1962theory},
and we include a proof in \Cref{sec:dom}.

\begin{lemma}[Dominating Sets in Dense Graphs \cite{blank1973estimate,mccuaig1989domination}] \label{lem:dominating-set-2}
If $G$ has $n\ge 8$ nodes, and every vertex in $G$ has degree at least two,
then $G$ has a dominating set of size at most $2n/5$.
\end{lemma}

\begin{lemma}[Dominating Sets in Regular Graphs  \cite{arnautov1974estimation,payan1975nombre}]
    \label{lem:regular-dom}
    Any $d$-regular graph on $n$ vertices has a dominating set of size at most $\grp{H_{d+1}/(d+1)}n$.
\end{lemma}

\begin{restatable}[Dominating Set Algorithm]{lemma}{domsetalgorithm}
\label{lem:dominating-set}
    Let $G'$ be a graph obtained by starting with $G$, and repeatedly deleting vertices of degree one until no vertices of degree one remain. 
    Suppose at most $n/5$ unit-degree vertices were deleted from $G$ to produce $G'$.
    Then a minimum-size dominating set of $G'$ can be found in $O^*(2^{m - 3n/5})$ time and polynomial space.
\end{restatable}

\Cref{lem:dominating-set} is proved in \Cref{sec:dom}.

\paragraph*{Finite Field Arithmetic}

Throughout, we work with polynomials and matrices over a field $\mathbb{F}=\FF_{2^\ell}$ of characteristic two, where $\ell = \poly(n)$ is sufficiently large. 
Arithmetic operations over $\FF$ take $\poly(n)$ time. 
All elements $a\in\FF$ satisfy $a^{2^\ell} = a$.
Consequently, 
given any element $a\in\FF$, we can compute its unique square root as  $a^{2^{\ell-1}}$ by repeated squaring in $\poly(n)$ time.

\paragraph*{Polynomials}
Let $P(X)$ be a polynomial over a set of variables $X = \{ x_1, \cdots, x_n \}$.
A monomial $\mathbf{m}$ 
is a product $\mathbf{m} = x_1^{m_1} \cdots x_n^{m_n}$ for some nonnegative integers $m_1, \cdots, m_n$.
A monomial $\mathbf{m}$ is \emph{multilinear} if $m_i \le 1$ for each $i \in [n]$.
We say $\monomial$ appears in the polynomial $P$ if the coefficient of $\monomial$ in $P$ is nonzero.
We define the \emph{support} of $\mathbf{m}$
to be the set of indices
\[\supp(\mathbf{m}) = \{ i \in [n] \mid m_i
> 0 \}\]
of variables which appear with positive degree in $\mathbf{m}$.

Similarly, we define the \emph{odd support} 
of $\mathbf{m}$ to be the set 
\[\osupp(\mathbf{m}) = \{ i \in [n] \mid m_i \equiv 1 \mod 2 \}\]
of indices of variables which appear with odd degree in 
$\mathbf{m}$.
By definition, $\osupp(\monomial) \subseteq \supp(\monomial)$ for all monomials $\monomial$.

The \emph{degree} of a monomial $\monomial = x_1^{m_1}\cdots x_n^{m_n}$ is defined to be $\deg(\monomial) = m_1 + \dots + m_n.$
Given  $S\sub [n]$, we define the degree of $\monomial$ restricted to the variables indexed by $S$ to be 
    \[\sum_{i\in S} m_i.\]
The \emph{total degree} (or just \emph{degree}) of $P$ is defined to be 
\[\deg P = \max_{\monomial} \, \deg(\monomial) \]
where the maximum is taken over all monomials $\monomial$ with nonzero coefficient in $P$.
We say that a polynomial $P$ is \emph{homogeneous} if $\deg(\monomial) = \deg P$ for all monomials $\monomial$ in $P$.

We recall the Lagrange Interpolation formula, which allows one to recover coefficients of a low degree univariate polynomial from a small number of evaluations of that polynomial. 

\begin{proposition}[Lagrange Interpolation]
  \label{lemma:interpolation}
  Let $P(z)$ be a univariate polynomial of degree less than $n$. 
  Suppose that $P(z_i) = p_i$ for distinct $z_1, \cdots, z_n \in \F$.
  Then 
  \[
    P(z)
    = \sum_{i \in [n]} p_i \prod_{j \in [n] \setminus \{ i \}} \frac{z - z_j}{z_i - z_j}.
  \]
  So given $n$ evaluations of $P$ at distinct points,
  we can compute
  all the coefficients of $P(z)$ in polynomial time.
\end{proposition}

We also record the following observation, proven for example in \cite[Theorem 7.2]{Motwani1995}, which shows that to test whether a low degree polynomial $P$ is nonzero, it suffices to check whether a random evaluation of $P$ over a sufficiently large field is nonzero. 

\begin{proposition}[Schwartz-Zippel Lemma]
\label{schwartz-zippel}
    Let $P$ be a nonzero polynomial over a finite field $\FF$ of degree at most $d$.
    If each variable of $P$ is assigned an independent, uniform random value from $\FF$,
    then the corresponding evaluation of $P$ is nonzero with probability $1-d/|\FF|$.
\end{proposition}

\paragraph*{Matrices}

Given a matrix $A$ and sets of rows $I$ and columns $J$, we let $A[I, J]$ denote the submatrix of $A$ restricted to rows in $I$ and columns in $J$. 
We let $A[I,\cdot]$ denote the submatrix of $A$ restricted to rows in $I$ but using all columns, and $A[\cdot,J]$ denote the submatrix of $A$ restricted to columns in $J$ but using all rows. 
If the rows and column sets of $A$ are identical, we write $A[S]$ as shorthand for $A[S, S]$.
We let $O$ denote the all-zeros matrix, whose dimensions will always be clear from context. 
We let $A^\top$ denote the transpose of $A$.
A matrix $A$ is \emph{symmetric} if $A=A^\top$, and \emph{skew-symmetric} if it is symmetric and has zeros along its main diagonal (we employ this definition because we work over characteristic two).

If $A$ is a square matrix, we let $\det A$ denote the determinant of $A$.

Given a set $V$, 
a perfect matching on $V$ is a partition of $V$ into sets of two elements each.
We let $\Pi(V)$ denote the set of perfect matchings on $V$.
Given a skew-symmetric matrix $A$ with rows and columns indexed by a set $V$, we define the \emph{Pfaffian} of $A$ to be 
\begin{equation}
\label{eq:pfaffian-def}
  \Pf A = \sum_{M\in \Pi(V)} \prod_{\set{u,v}\in M} A[u, v].
\end{equation}
The standard definition for Pfaffians involves a sign term
(see e.g., \cite[Section 7.3.2]{murota1999matrices}),
which we ignore here because we work over a field of characteristic two. 
It is well known that for all skew-symmetric $A$, we have  $\det A = (\Pf A)^2$ (see e.g., \cite[Proposition 7.3.3]{murota1999matrices} for a proof).
Consequently, the Pfaffian of matrix $A$ over the field $\FF_{2^\ell}$ can be computed in $\poly(n,\ell) \le \poly(n)$ time by using the equation $\Pf A = \sqrt{\det A} = (\det A)^{2^{\ell-1}}$.

We record the following useful  facts for computation with Pfaffians and determinants.

\begin{restatable}[Direct Sums]{proposition}{pfaffiansum} \label{lemma:direct-sum}
For any skew-symmetric matrices $A_1$ and $A_2$, we have 
  \begin{align*}
    \Pf \begin{pmatrix}
       A_1 & O \\ O & A_2 
    \end{pmatrix} = \grp{\Pf A_1 }\grp{\Pf A_2}
  \end{align*}
over a field of characteristic two.
\end{restatable}

\Cref{lemma:direct-sum} is proved in 
\Cref{sec:construction}.

\begin{restatable}[Ishikawa-Wakayama Formula]{proposition}{iwformula}
  \label{lemma:ishikawa-wakayama}
  For a skew-symmetric $n \times n$-matrix $A$ and a $k \times n$-matrix $B$ with $k \le n$, we have
  \begin{align*}
    \Pf B A B^\top = \sum_{S\in \binom{[n]}{k}} \det B[\cdot, S] \Pf A[S].
  \end{align*}
\end{restatable}

\Cref{lemma:ishikawa-wakayama} is due to \cite{IshikawaW95} (see  \cite[Section 4.3] {koana2024faster} for a recent alternate proof).

\paragraph*{Matroids}

A \emph{matroid} is a pair $\matroid = (X, \cI)$, consisting of a \emph{ground set} $X$ and a collection $\cI$ of subsets of $X$ called  \emph{independent sets}, satisfying the following axioms:
\begin{enumerate}
    \item $\emptyset \in \cI$.
    \item If $B \in \cI$ and $A \subset B$, then $A \in \cI$.
    \item For any $A, B \in \cI$ with $|A| < |B|$, there exists an element $b \in B \setminus A$ such that $A\cup\set{b} \in \cI$.
\end{enumerate}

A maximal independent set in a matroid $\matroid$ is called a \emph{basis}.
It is well known that all bases of $\matroid$ have the same size $r$, which is referred to as the \emph{rank} of $\mcal{M}$. 
We say $S\sub V$ \emph{spans} $\mcal{M}$ if $S$ is a superset of a basis of $\mcal{M}$.
If there exists a matrix $A$ with column set $X$ 
such that for every $S\sub X$, 
the set $S$ is independent in $\matroid$ if and only if $A[\cdot, S]$ has full rank, then we say that $\matroid$ is a \emph{linear} matroid \emph{represented} by $A$. 
Provided we work over a large enough field of size $\poly(|X|)$, 
it is known that if $\matroid$ is linear,
then it can be represented by a matrix whose number of rows equals the rank of $\matroid$ (see e.g., \cite{LokshtanovMPS18TALG} or \cite[Section 3.7]{Marx09-matroid}).

    Given matroids $\matroid_1, \dots, \matroid_k$ over disjoint ground sets $X_1, \dots, X_k$ respectively,
    we define the direct sum
        \[\matroid = \bigoplus_{i=1}^k \matroid_i\]
    of these matroids
    to be the ground set $X = X_1\sqcup \dots \sqcup X_k$ together with the family $\cI$ of subsets $S$ of $X$ with the property that $S\cap X_i$ is independent in $\matroid_i$ for all $i\in [k]$.
    It is well known that $\matroid$ is a matroid as well, whose rank is equal to the sum of the ranks of the $\matroid_i$.

    The following result lets us represent direct sums of linear matroids \cite[Section 3.4]{Marx09-matroid}.

\begin{proposition}[Matroid Sum]
\label{prop:matroid-sum}
    Given linear matroids $\matroid_1, \dots, \matroid_k$ represented by matrices $A_1, \dots, A_k$ respectively,
    their direct sum is represented by the block-diagonal matrix 
        \[
  A = \begin{pmatrix}
    A_1 & O & \dots & O \\
    O & A_2 & \dots & O \\
    \vdots & \vdots & \ddots & \vdots \\
    O & O & \dots & A_k
  \end{pmatrix}.
        \]
\end{proposition}

In our \edgecolor{} algorithm, we make use of two simple types of matroids: uniform and partition matroids.
Given a ground set $X$ and a nonnegative integer $r$, 
the \emph{uniform matroid} over $X$ of rank $r$ has as its independent sets all subsets of $X$ of size at most $r$.

\begin{proposition}[Uniform Matroid Representation {\cite[Section 3.5]{Marx09-matroid}}]
\label{prop:uniform}
    Given a set $X$, a nonnegative integer $r$, and a field $\FF$ with $|\FF| > |X|$,
    we can construct an $r\times |X|$ matrix representing the uniform matroid over $X$ of rank $r$ in $\poly(|X|)$ time. 
\end{proposition}

Given a set $X$ partitioned $X = X_1\sqcup \dots \sqcup X_k$ into parts $X_i$, and a list of nonnegative integer capacities $c_1, \dots, c_k$, 
the \emph{partition matroid} $\matroid$ over the ground set $X$ for this partition and list of capacities has as its independent sets all subsets $S\sub X$ satisfying
    $|S\cap X_i| \le c_i$
for all $i\in [k]$.
Equivalently, $\matroid$ is the direct sum of uniform matroids of rank $c_i$ over $X_i$ for all $i\in [k]$.
By definition, $\matroid$ has rank $r = (c_1 + \dots + c_k)$.
By \Cref{prop:matroid-sum,prop:uniform} we see that we can efficiently construct representations of partition matroids
\cite[Proposition 3.5]{Marx09-matroid}.

\begin{proposition}[Partition Matroid Representation]
\label{prop:partition}
    Given a partition matroid $\matroid$ over $X$ of rank $r$, and a field $\FF$ with $|\FF| > |X|$,
    we can construct an $r\times |X|$ matrix representing $\matroid$ in $\poly(|X|)$ time. 
\end{proposition}

\section{Polynomial Sieving} \label{sec:polynomial-sieving}

In this section, we develop faster algorithms for a certain monomial detection problem.
This is our main technical result, which allows us to obtain faster algorithms for \edgecolor{} and \textsc{List Edge Coloring}.
We recall the following well-known sieving technique:

\begin{proposition}[Inclusion-Exclusion]
  \label{lemma:inclusion-exclusion}
  Let $P(X)$ be a polynomial over a field of characteristic two, with variable set $X = \{ x_1, \cdots, x_n \}$.
    Fix $T \subseteq X$.
    Let $Q(X)$ be the polynomial consisting of all monomials (with corresponding coefficients) in $P$ that are divisible by $\prod_{x \in T} x$.
  Then 
  \[Q = \sum_{S \subseteq T} P_{S}\] where 
  $P_{S}$ is the polynomial obtained from $P$ by setting $x_i = 0$ for all $i\in S$.
\end{proposition}

See \cite[Lemma 2]{Wahlstrom13STACS} for a proof of \Cref{lemma:inclusion-exclusion}.

Combining \Cref{lemma:interpolation,lemma:inclusion-exclusion}, we obtain the following result.

\begin{lemma}[Coefficient Extraction]
\label{lemma:coefficient-extraction}
  Let $P(X)$ be a polynomial over a field of characteristic two, with variable set $X = \{ x_1, \cdots, x_n \}$.
  For $T \subseteq X$, let $Q(X \setminus T)$ be the coefficient of $\prod_{x \in T} x$ in $P(X)$,
  viewed as a polynomial over the variable set $X\setminus T$.
  Then we can evaluate $Q(X \setminus T)$ at a point as a linear combination of $2^{|T|}\poly(n)$ evaluations of $P(X)$.
\end{lemma}
\begin{proof}
    Fix a subset $T\sub X$.
    Let $F$ be the polynomial with variable set $X\cup\set{z}$ obtained by
    substituting each variable $x\in T$ with $xz$ in $P$. 
    Let $G(X)$ be the coefficient of $z^{|T|}$ in $F$.
    By definition of $F$, $G$ is the polynomial obtained by taking the monomials in $P$ whose degree restricted to $T$ equals $|T|$.
    By \Cref{lemma:interpolation}, we can compute $G(X)$ at any point as a linear combination of $\poly(n)$ evaluations of $P$.

    By definition, $Q(X\setminus T)$ is the polynomial obtained by taking monomials of $G(X)$ that are divisible by $\prod_{x\in T} x$, and then setting all variables in $T$ equal to $1$.
    So by \Cref{lemma:inclusion-exclusion} we can evaluate $Q(X\setminus T)$ as the sum of $2^{|T|}$ evaluations of $G$.
    Then by the conclusion of the previous paragraph, we get that we can evaluate $Q(X\setminus T)$ at any point as a linear combination of $2^{|T|}\poly(n)$ evaluations of $P$ as claimed. 
\end{proof}

The inclusion-exclusion principle allows us to sieve for the multilinear term in a polynomial that is the product of all its variables.
This was extended to the parameterized setting, where the goal is to sieve for multilinear terms of degree $k$ in $O^*(2^k)$ time \cite{Bjorklund14detsum,BjorklundHKK17narrow,Williams09IPL}, and then
further extended to multilinear monomial detection in the matroid setting \cite{determinantalsieving}.

\begin{lemma}[Basis Sieving {\cite[Theorem 1.1]{determinantalsieving}}] \label{lemma:basis-sieving}
  Let $P(X)$
  be a homogeneous polynomial of degree $k$ over a field $\F$ of characteristic~2, 
  and let $\mcal{M}=(X,\cI)$ be a matroid on $X$ of rank~$k$, with a representation over $\F$ given. 
  Then there  is a randomized algorithm
  running in $O^*(2^k)$ time and polynomial space, which determines with high probability if $P(X)$ contains a multilinear monomial $\mathbf{m}$ such that $\supp(\monomial)$ is a basis of $\mcal{M}$.
\end{lemma}

Moreover, \cite{determinantalsieving} introduced a variant of sieving that takes the odd support into account.
This tool is key to our exponential speed-up for the \edgecolor{} problem.

\begin{lemma}[Odd Sieving {\cite[Theorem 1.2]{determinantalsieving}}] \label{lemma:odd-sieving}
  Let $P(X)$
  be a polynomial of degree $d$ over a field $\F$ of characteristic~2, 
  and let $\mcal{M}=(X,\cI)$ be a matroid on $X$ of rank $k$, represented by a $k\times n$ matrix over $\FF$. 
  Then using $O^*(d2^k)$ evaluations of $P(X)$ and polynomial space, we can
  determine if $P(X)$ contains a monomial $\monomial$ such that $\osupp(\monomial)$ spans $\mcal{M}$.
\end{lemma}

Let $P(X)$ be a homogeneous polynomial of degree $d$ in variables $X = \{ x_1, \dots, x_n \}$.
Consider a partition $X = X_1\sqcup \dots \sqcup X_p$ into parts $X_i$.
Let $c_1, \dots, c_p\ge 0$ be integers with 
    \[c_1 + \dots + c_p = d.\]
Let $\mcal{M}$ be the partition matroid defined by the partition of $X$ into the parts $X_i$ and the capacities $c_i$ for $i\in [p]$.
We say $P(X)$ is \emph{compatible} with $\mcal{M}$ if for every $i\in [p]$ and every monomial $\monomial$ in $P$, the degree of $\monomial$ restricted to $X_i$ equals $c_i\ge 1$.

The following result gives an improvement over \cref{lemma:basis-sieving} for 
polynomials that are compatible with partition matroids.

\begin{theorem}[Partition Sieving]
\label{theorem:faster-sieving}
  Let $P(X)$
  be a polynomial of degree $d$ over a field $\F$ of characteristic~2 with $|\F| > d$,
  and let $\mcal{M}=(X,\cI)$ be a partition matroid on $X$ of rank $d$ with $p$ partition classes.
  If $P$ is compatible with $\mcal{M}$,
  then there is a randomized algorithm
  which runs in $O^*(2^{d - p})$ time and polynomial space that determines with high probability if
  $P(X)$ contains a monomial $\monomial$ such that $\supp(\monomial)$ is a basis of $\mcal{M}$.
\end{theorem}
\begin{proof}
Let $X = X_1\sqcup \dots \sqcup X_p$ be the partition and $c_1, \dots, c_p$
be the capacities defining the partition matroid $\matroid$.
    By assumption, $c_i\ge 1$ for all $i\in [p]$.
    Let $\matroid'$ be the partition matroid over $X$ defined by the same partition as $\matroid$, but with capacities $(c_i-1)$ for $i\in [p]$.
    From this construction, $\matroid'$ has rank $(d-p)$.

  Our algorithm to determine whether $P(X)$ contains a monomial whose 
  support spans $\matroid$ works as follows.
  By \Cref{prop:partition}, a linear representation of $\matroid$ over $\F$ can be computed in polynomial time.
  We run the odd sieving algorithm of \Cref{lemma:odd-sieving} with $P$ and $\mcal{M}'$ 
  to determine
  in $O^*(2^{d - p})$ time
  whether $P$ contains a monomial whose odd support spans $\matroid'$.
  We return YES if such a monomial exists, and return NO otherwise.
  We now prove that this algorithm is correct.

    First, suppose that $P(X)$ has a monomial $\monomial$ whose support $\supp(\monomial)$ is a basis of $\matroid$.
    By the compatibility condition, $\monomial$ is multilinear, which implies that $\supp(\monomial) = \osupp(\monomial)$.
    Since $\osupp(\monomial)$ is a basis of $\matroid$, and thus spans $\matroid$, it follows that $\osupp(\monomial)$ spans $\matroid'$.

  Conversely, suppose $P(X)$ has a monomial $\monomial$
  such that $\osupp(\monomial)$ spans $\mcal{M}'$.
  We claim that $\supp(\monomial)$ spans $\mcal{M}$. 
Indeed, since $\osupp(\monomial)$ spans $\matroid'$,
the set $\osupp(\monomial)$
contains at least $(c_i - 1)$ variables of $X_i$ for each $i\in [p]$.
Hence its superset $\supp(\monomial)$  contains at least $(c_i-1)$ variables of each part $X_i$ as well.

\begin{claim}
\label{claim:linear}
    The set $\supp(\monomial)$ contains exactly $c_i$ variables of each part $X_i$.
\end{claim}
\begin{claimproof}
    Suppose to the contrary that $\supp(\monomial)$
    does not contain $c_j$ variables of $X_j$ for some index $j\in [p]$.
    Then $\monomial$ contains exactly $(c_{j}-1)$ variables of $X_j$.
    Since $P$ is compatible with $\matroid$, the monomial $\monomial$ has degree exactly $c_j$ when restricted to $X_j$.
    The only way this is possible is if $c_j \ge 2$,
    and $\monomial$ has exactly $(c_j-2)$ variables in $X_j$ of degree one and a single variable in $X_j$ with degree two. 
    The variable of degree two does not show up in the odd support of $\monomial$,
    which implies that $\osupp(\monomial)$ has at most $(c_j-2)$ elements,
    which contradicts the assumption that $\osupp(\monomial)$ spans $\matroid{}'$.
\end{claimproof}

By \Cref{claim:linear}, the set $\supp(\monomial)$
has exactly $c_i$ variables from $X_i$
for each $i\in [p]$.
Since $P$ is compatible with $\matroid$,
this implies that $\monomial$  has degree one in every variable in $X$. 
Hence $\supp(\monomial)=\osupp(\monomial)$ spans $\matroid$, as claimed.

This shows that the algorithm is correct, and proves the desired result. 
\end{proof}

\section{Coloring Algorithm Template}
\label{sec:template}

Recall that in the \textsc{List Edge Coloring}
problem we are given a graph $G = (V, E)$, an integer $k \ge 1$, and lists $L_e \sub [k]$ of colors for every edge $e \in E$, and are tasked with determining
if $G$ contains an assignment of colors $c(e)\in L_e$ to each edge such that no two edges incident to the same vertex are assigned the same color.
We call such an assignment a proper edge coloring of $G$ with respect to the lists $L_e$, or just a proper list edge coloring if the $L_e$ are clear from context. 

For each $i \in [k]$, let $G_i = (V, E_i)$ be the subgraph of $G$, where we define $E_i\sub E$ to be the set of all edges $e$ in the graph which have $i \in L_e$ available as a color.
Let $\mcal{F}$
denote the collection of $k$-tuples $(M_1, \dots, M_k)$  of matchings such that 
\begin{enumerate}
    \item for all $i\in [k]$, $M_i$ is a matching in $G_i$, and 
    \item 
    for all vertices $v\in V$, $v$ appears as the endpoint of exactly $\deg_G(v)$ of the $M_i$ matchings.
\end{enumerate} 

Observe that $G$ admits a proper list edge coloring if and only if $\mathcal{F}$ contains a tuple of edge-disjoint matchings.
Indeed, suppose $G$ admits a proper list edge coloring. 
For each $i \in [k]$, the set of edges assigned color $i$ must be a matching $M_i\sub E_i$ in $G$.
Moreover, every edge is assigned a color, so each vertex $v$ appears as an  endpoint in $\deg_G(v)$ of these matchings,
hence the matchings form an edge-disjoint tuple in $\mathcal{F}$.
Conversely, given an edge-disjoint tuple $(M_1, \dots, M_k)\in\mathcal{F}$, assigning the edges in $M_i$ color $i$ gives a color to every edge in $G$ by condition 2 above, and thus yields a proper list edge coloring by condition 1. 

The basic idea of our algorithmic template is to construct a polynomial $P$ which enumerates a certain subfamily $\mathcal{C}\sub\mathcal{F}$.
We then argue that $P$ contains a multilinear monomial satisfying certain matroid constraints if and only if $\mathcal{C}$ contains a tuple of edge-disjoint matchings.
Using similar reasoning to the previous paragraph, this then lets us solve \listedgecolor{} by running monomial detection algorithms on $P$.

\subsection{Enumerating Matchings}
\label{ssec:polynomial}

In this section, we design a polynomial $P$ whose monomials enumerate tuples of matchings satisfying special conditions. 
To construct $P$, we first define certain polynomial matrices.

For each edge $e \in E$ we introduce a variable $x_e$, and for each
 pair $(e,i)\in E\times [k]$
we introduce a variable  $y_{ei}$.
Let $X$ be the set of $x_e$ variables and $Y$ be the set of $y_{ei}$ variables. 

For each $i\in [k]$, define $A_i$ to be the matrix  with rows and columns indexed by $V$ with
\begin{align*}
  A_i[u,w] = A_i[w, u] = \begin{cases}
    x_{e} y_{ei} &  \text{ if } e = \set{u,w}\in E_i, \\
    0 & \text{ otherwise}. \\
  \end{cases}
\end{align*}

For each $i \in [k]$, define $V_i = \{ v_i \mid v \in V \}$ to be a copy of $V$.
Let 
    \[W = \bigsqcup_{i=1}^k V_i\]
be the union of these copies. 
Let $A$ be the symmetric matrix with rows and columns indexed by $W$, defined by taking
\begin{equation}
\label{eq:A-def}
  A = \begin{pmatrix}
    A_1 & O & \dots & O \\
    O & A_2 & \dots & O \\
    \vdots & \vdots & \ddots & \vdots \\
    O & O & \dots & A_k
  \end{pmatrix}
\end{equation}
and identifying $V_i$ as the row and column set indexing $A_i$.

Note that $A$ is $(kn)\times (kn)$, and over a field of characteristic two is skew-symmetric.

For each vertex $v\in V$, let $S_v = \set{v_i\mid i\in [k]}$
be the set of $k$ copies of $v$. 

For each vertex $v$,
let $\matroid_v$ be a linear matroid of rank $\deg_G(v)$ over $S_v$.
These matroids will be specified later on, in our algorithms for \edgecolor{} and \listedgecolor{}.
Let $\mcal{W}$ be the direct sum 
    \begin{equation}
    \label{eq:W-matroid}
    \watroid = \bigoplus_{v\in V} \matroid_v
    \end{equation}
of these matroids. 

By \Cref{lemma:direct-sum} the matroid $\mcal{W}$ is linear and has rank equal to the sum of the degrees of vertices in $G$, which is $2m$.
Let $B$ be a $2m\times |W|$ matrix  representing $\mcal{W}$.

Let $\cal{C}$ be the collection of $k$-tuples of matchings $(M_1, \dots, M_k)$ in $G$ such that 
\begin{enumerate}
    \item all edges in $M_i$ have lists containing the color $i$, 
    \item every vertex $v$ is incident to exactly $\deg_G(v)$ edges across the matchings $M_i$, and 
    \item for each vertex $v$, the set of copies $v_i$ taken over all $i$ such that $v$ appears as an endpoint of an edge in $M_i$ forms an independent set in $\matroid_v$.
\end{enumerate}

\noindent We define the polynomial
\begin{equation}
\label{eq:matching-polynomial}
P(X,Y) = \Pf BAB^\top.
\end{equation}
The next result shows that $P$ enumerates tuples of matchings from $\cal{C}$.

\begin{lemma}[Enumerating Matchings]
\label{lem:generating-function}
We have 
  \begin{align*}
    \Pf B A B^\top = \sum_{(M_1,\dots,M_k)\in\mcal{C}} c_{M_1,\dots,M_k} \left( \prod_{i \in [k]} \prod_{e \in M_i} y_{ei} \right) \left( \prod_{e \in M_1 \cup \dots \cup M_k} x_e \right),
  \end{align*}
  where each $c_{M_1,\dots,M_k} \ne 0$ is a constant  depending only on $M_1,\dots, M_k$.
\end{lemma}
\begin{proof}
  Let $U \subseteq W$ be a subset of $W$ of size $2m$.
  By \Cref{lemma:ishikawa-wakayama}, we have
  \begin{equation}
  \label{eq:first-iw}
    \Pf B A B^\top
    = \sum_{U \in \binom{W}{2m}} \det B[\cdot, U] \Pf A[U].
  \end{equation}

    By \cref{eq:A-def} and \Cref{lemma:direct-sum}, for every subset $U\sub W$ of size $2m$ we have 
    \[ \Pf A[U]
    = \prod_{i \in [k]} \Pf A_i[U \cap V_i].\]

    Substituting the above equation into \Cref{eq:first-iw} yields
        \[
    \Pf B A B^\top  = \sum_{U \in \binom{W}{2m}} \det B[\cdot, U] \prod_{i \in [k]} \Pf A_i[U \cap V_i].
        \]

  Recall that given a graph $H$,  we let $\Pi(H)$
    denote the set of perfect matchings of $H$.
  By expanding out the definition of $\Pf A_i[U \cap V_i]$ in the above equation, we see that 
  \begin{align*} 
    \Pf B A B^\top = \sum_{U \in \binom{W}{2m}} \det B[\cdot, U] \prod_{i \in [k]} \sum_{M_i\in \Pi(G[U_i])} \prod_{e \in M_i} A_i[e],
  \end{align*}
  where $U_i \subseteq V$ is defined as a copy of $U \cap V_i$, i.e., $U_i = \{ v \mid v_i \in U \cap V_i \}$.
  By interchanging the order of the product and the summation, we see that the above is equivalent to summing over all collections $(M_1, \dots, M_k)$, where $M_i$ is a matching in $G[U_i]$ for each $i \in [k]$:
    \begin{align*} 
    \Pf B A B^{\top} = \sum_{U \in \binom{W}{2m}} \det B[\cdot, U] \sum_{(M_1, \dots, M_k)} \prod_{i \in [k]} \prod_{e \in M_i} A_i[e].
  \end{align*}

    Let $\vec{M} = (M_1, \dots, M_k)$ be a tuple such that the summand corresponding to $\vec{M}$ in the above equation appears with nonzeo coefficient in $\Pf BAB^\top$.
    Since $A_i[e]$ is nonzero if and only if $i\in L_e$, the tuple $\vec{M}$ meets condition 1 in the definition of $\cal{C}$.
  By the definition of $B$, we have $\det B[\cdot, U] \ne 0$ if and only if $U\cap S_v$ is independent in $\matroid_v$ for every vertex $v \in V$.
  Thus $\vec{M}$ meets  condition 3 the definition of $\cal{C}$.
Since $\matroid_v$ has rank $\deg_G(v)$ for each vertex $v$, $U$ has size $2m$, and the sum of degrees of vertices in $G$ is $2m$,
we see that
the only way for $U\cap S_v$ to be independent in $\matroid_v$ for all  $v$ is if in fact $|U\cap S_v| = \deg_G(v)$ always.  
  So $\vec{M}$ satisfies condition 2 from the definition of $\cal{C}$.
  
  Thus, the nonzero terms in the above sum correspond precisely to tuples $\vec{M}\in \cal{C}$,
  which proves the desired result.
\end{proof}

\subsection{Partition Sieving With Dominating Sets}
\label{ssec:sieving}

In this section, we describe a generic way of going from a dominating set in $G$ to a certain partition matroid $\datroid$.
After defining this $\datroid$, we will apply \Cref{theorem:faster-sieving} to $P$ and $\datroid$ to achieve a fast algorithm for detecting a multilinear monomial in $P$.

Recall that a \emph{dominating set} of a graph $G$ is a subset of vertices $D$ such that every vertex in $G$ is either in $D$ or adjacent to a vertex in $D$.

Fix a dominating set $D\sub V$ of the input graph $G$.
Let $\notD = V\setminus D$.
Let $E'$ be the set of edges in $G$ with one endpoint in $D$ and one endpoint in $\notD$,
and $X' = \set{x_e\mid e\in E'}$ be the corresponding set of variables.
For each  $v\in V'$, let 
    \[\partial(v) = \set{e\in E' \mid e\ni v}\]
be the set of edges incident to $v$ which connect $v$ to a vertex in $D$. 
By definition, we have a partition
    \[E' = \bigsqcup_{v\in V'} \partial(v).\]
Let $\datroid$ be the partition matroid over $E'$ with respect to the above partition,
and with capacities $c_v = \deg_{E'}(v)$ for each $v\in V'$,
where $E'$ is viewed as a subgraph of $G$. 
Since $D$ is a dominating set, every vertex $v\in V'$ is adjacent to some node in $D$, and thus the capacities $c_v\ge 1$ are all positive. 

We now use this partition matroid to sieve over the polynomial $P$, defined in \Cref{eq:matching-polynomial}.

\begin{lemma}[Accelerated Multilinear Monomial Detection]
    \label{lem:template}
    There is a randomized algorithm that determines with high probability whether $P(X,Y)$
    has a monomial divisible by $\prod_{e\in E} x_e$, running in $O^*(2^{m-|\notD|})$ time and polynomial space. 
\end{lemma}
\begin{proof}
    Let $Q(X',Y)$ be the coefficient of 
        \[\prod_{e\in E\setminus E'} x_e\]
    in $P(X,Y)$.
    By definition, $P(X,Y)$ contains a monomial divisible by $\prod_{e\in E} x_e$
    if and only if $Q(X',Y)$ contains a monomial divisible by $\prod_{e\in E'} x_e$.

Let $R(Y)$ be the coefficient of $\prod_{e\in E'} x_e$  in $Q$, viewed as a polynomial with variable set~$Y$. 
Note that the total degree of $R(Y)$ is $\poly(n)$.
Take a uniform random assignment over $\FF$ to the variables in $Y$. 
By \Cref{schwartz-zippel}, if $R$ is nonzero, then it remains nonzero with high probability after this evaluation, provided we take $|\FF| = \poly(n)$ to be sufficiently large. 
For the rest of this proof, we work with polynomials $P, Q, R$ after this random evaluation, so that $P$ is a polynomial in $X$, $Q$ is a polynomial in $X'$,
and $R$ is a field element. 

By the discussion in the first paragraph of this proof,
it suffices to check whether $Q(X')$ contains the monomial $\prod_{e\in E'} x_e$ in $O^*(2^{m-|V'|})$ time and polynomial space.

    We observe the following helpful property of $Q$. 

    \begin{claim}
    \label{lem:compatibility}
        The polynomial $Q(X')$ is compatible with the matroid $\datroid$.
    \end{claim}
    \begin{claimproof}
Take an arbitrary monomial $\monomial$ in $Q$.
By definition, there is a corresponding monomial 
\[\monomial' = \monomial\cdot \prod_{e \in E \setminus E'} x_e\]
in $P$.
By \Cref{lem:generating-function}, 
for every $v \in V$, $\monomial'$ has degree exactly $\deg_G(v)$ when restricted to the variables $x_e$ corresponding to edges $e$ containing $v$.
Consequently, for all $v \in V'$, $\monomial$ has degree exactly 
\[\deg_G(v) - |\{ e\not\in E' \mid e\ni v \}| = \deg_{E'}(v)\] 
when restricted to  variables $x_e$ with $e \in \partial(v)$.    
Since this holds for arbitrary monomials $\monomial$ in $Q$, the polynomial $Q$ is compatible with the partition matroid $\datroid$ as claimed. 
    \end{claimproof}

By \Cref{lem:generating-function}, $P$, viewed as a polynomial in $X$, is homogeneous of degree $|E| = m$. 
Hence $Q$, viewed as a polynomial in $X'$,
is homogeneous of degree $d = |E'|$. 
By \Cref{lem:compatibility}, $Q$ is compatible with $\datroid$.
Since $\datroid$ is a partition matroid with positive capacities and $p = |V'|$ parts in its underlying partition, by 
 \Cref{theorem:faster-sieving}
 we can determine whether $Q$ contains a monomial $\monomial$ in the $X'$ variables such that $\supp(\monomial)$ is a basis of $\datroid$ using
    \begin{equation}
    \label{eq:Q-eval-bound}
    O^*(d2^{d-p}) \le O^*(2^{|E'|-|V'|})
    \end{equation}
evaluations of $Q$.
Since $E'$ is the unique basis of $\datroid$,
this means we can test whether $Q(X')$ contains the monomial $\prod_{e\in E'} x_e$ using $O^*(2^{|E'| - |V'|})$ evaluations of $Q(X')$.

    By \Cref{lemma:coefficient-extraction}, $Q(X')$ can be evaluated at any point by computing a linear combination of $O^*(2^{m - |E'|})$ evaluations of $P(X)$.
    Combining this observation with \Cref{eq:Q-eval-bound}, we get that we can determine whether $Q(X')$ contains the monomial $\prod_{e\in E'}x_e$
    by computing a linear combination of 
    \[O^*(2^{m-|E'|} \cdot 2^{|E'|-|V'|}) \le O^*(2^{m - |V'|})\]
    evaluations of $P(X,Y)$.
    Each evaluation of $P(X,Y) = \Pf BAB^\top$ takes polynomial time, so the desired result follows. 
\end{proof}

\section{Edge Coloring Algorithm}
\label{sec:edge-color}

In this section, we present our algorithm for \edgecolor{} and prove \Cref{thm:edgecolor,thm:regular}.
Recall that in this problem, we are given an input graph $G$ with $n$ vertices, $m$ edges, and maximum degree $\Delta$.
To solve the problem, it suffices to determine whether $G$ admits a proper edge coloring using at most $\Delta$ colors (i.e., whether $\chi'(G)\le \Delta$).

\subsection{Removing Low Degree Vertices}

\begin{lemma}[Unit-Degree Deletion]
 \label{lem:delete-deg-1}
Let $G$ be a graph with a vertex $v$ of degree one. 
Let $G'$ be the graph obtained by deleting $v$ from $G$. 
Then $\chi'(G')\le \Delta$ if and only if $\chi'(G)\le \Delta$.
\end{lemma}
\begin{proof}
    If $G$ admits a proper edge coloring with $\Delta$ colors, then this coloring is also a proper edge coloring for $G'$ with $\Delta$ colors.

    Conversely, suppose $G'$ admits a proper edge coloring with $\Delta$ colors.
    Since $\deg_G(v) = 1$, there is a unique vertex $w$ in $G$ that $v$ is adjacent to.
    Then 
    \[\deg_{G'}(w) \le \deg_G(w) - 1 \le \Delta-1\]
    because $G$ has maximum degree $\Delta$.
    So vertex $w$ has edges using at most $(\Delta-1)$ distinct colors incident to it. 
    Taking the coloring of $G'$ and  assigning edge $\set{v,w}$ a color not used by the edges incident to $w$ in $G'$ recovers a proper edge coloring of $G$ with at most $\Delta$ colors.    
\end{proof}

\subsection{Instantiating the Template}

Recall the definitions from \Cref{sec:template}.
We view the \edgecolor{} problem as a special case of \listedgecolor{},
where $k=\Delta$ and every edge $e$ is assigned the list $L_e = [\Delta]$.
For each vertex $v$,
we set $\matroid_v$ to be the uniform matroid over $S_v$ of rank $\deg_G(v)$.
Then by \Cref{eq:W-matroid}, $\watroid$ is the partition matroid over $W$ defined by the partition
    \[W = \bigsqcup_{v\in V} S_v\]
and the capacities $c_v = \deg_G(v)$ for parts $S_v$.
By \Cref{prop:partition}, we can construct the matrix $B$ representing $\watroid$ in polynomial time.
In this case, condition 3 is identical to condition 2 in the definition of the collection $\cal{C}$.

\begin{lemma}[Characterization by Polynomials]
\label{lem:edgecolor-char}
    We have $\chi'(G)\le \Delta$ if and only if the polynomial $P(X,Y)$
    has a monomial divisible by $\prod_{e\in E} x_e$.
\end{lemma}
\begin{proof}
Suppose $P(X,Y) = \Pf BAB^\top$ contains a monomial divisible by $\prod_{e\in E} x_e$.
Since $P$ is homogeneous of degree $m$ in the $X$ variables, 
by \Cref{lem:generating-function} this means
there exists a $\Delta$-tuple of \emph{edge-disjoint} matchings $(M_1, \dots, M_{\Delta})\in \cal{C}$, such that every edge of $G$ appears in some matching $M_i$. 
Consequently, assigning the edges in $M_i$ color $i$ yields a proper edge coloring of $G$ using at most $\Delta$ colors.

Conversely, suppose $G$ has a proper edge coloring with $\Delta$ colors. 
Without loss of generality, let the set of colors used be $[\Delta]$.
For each $i\in [\Delta]$,
let $M_i$ be the set of edges assigned color $i$. 
Since this is a proper edge coloring, each $M_i$ is a matching in $G$. 
Since every edge is assigned a unique color, each vertex $v$ has exactly $\deg_G(v)$ edges across the matchings $M_i$.
Hence $(M_1, \dots, M_{\Delta})\in\cal{C}$.
Then by \Cref{lem:generating-function}, $P$ contains a monomial divisible by $\prod_{e\in E} x_e$.
\end{proof}

We can now prove \Cref{lem:domset} as a simple application of \Cref{lem:edgecolor-char}.

\domset*
\begin{proof}
    Let $G$ have vertex set $V$ and edge set $E$, and let $V' = V\setminus D$.
By \Cref{lem:edgecolor-char}, we have $\chi'(G)\le \Delta$ if and only if $P(X,Y)$ has a monomial divisible by $\prod_{e\in E} x_e$.
By \Cref{lem:template}, we can determine whether such a monomial exists in 
    \[O^*(2^{m - |V'|}) \le O^*(2^{m-n+|D|})\]
time and polynomial space, as desired. 
\end{proof}

Having established \Cref{lem:domset}, we are ready to present out \edgecolor{} algorithms.

\main*
\begin{proof}
To solve \edgecolor{},
it suffices to determine whether $\chi'(G)\le \Delta$.

If $G$ has a vertex of degree one, delete it.
Keep performing such deletions, until we are left with a graph containing no vertices of degree one.
By repeated application of \Cref{lem:delete-deg-1}, the resulting graph admits a proper edge coloring using $\Delta$ colors if and only if $\chi'(G)\le \Delta$.
If the resulting graph has constant size, we can determine this trivially.
Each such deletion decreases the number of edges and vertices in the graph by one, and thus also decreases the value of $(m-3n/5)$.
Consequently, to show the desired result, we may assume without loss of generality that $G$ has minimum degree two, and at least eight vertices. 

Since $G$ has minimum degree two and $n\ge 8$ nodes, by \Cref{lem:dominating-set-2}
the graph has a dominating set of size at most $2n/5$. 
By \Cref{lem:dominating-set}, we can then find a dominating set $D$ in $G$ with $|D|\le 2n/5$  in $O^*(2^{m-3n/5})$ time and polynomial space. 
Hence by \Cref{lem:domset} we can solve \edgecolor{} on $G$ in $O^*(2^{m-3n/5})$ time and polynomial space, as desired. 
\end{proof}

\regular*
\begin{proof}
Let $G$ be the input graph. 
    By exhaustive search over all subsets of vertices,
    we can find a minimum-size dominating set $D$ of $G$ in $O^*(2^n)$ time and polynomial space.
    Being $d$-regular, $G$ has $m=dn/2 \ge 3n$ edges, so we find $D$ in $O^*(2^{m-n})\le O^*(2^{m-\alpha_d n})$ time.
    
    Since $G$ is $d$-regular, by \Cref{lem:regular-dom} we must have $|D|\le n(H_{d+1}/(d+1))$.
    Then by \Cref{lem:domset} we can solve \edgecolor{} on $G$ in 
    \[O^*(2^{m-n+|D|}) \le O^*(2^{m-\alpha_d n})\] 
    time as claimed. 
\end{proof}

\section{List Edge Coloring Algorithm}
\label{sec:list-color}

In this section, we present our algorithm for \listedgecolor{} and prove \Cref{thm:main}.
Recall that in this problem, we are given a graph $G$ with $n$ vertices and $m$ edges, and for each edge $e$ have a  list $L_e\sub [k]$ of admissible colors for $e$. 
We are tasked with determining whether $G$ admits a proper edge coloring  where edges are assigned colors from their lists.

\subsection{Removing Low Degree Vertices}
\label{subsec:list-degree}

For \textsc{List Edge Coloring}, obtaining a statement analogous to \cref{lem:delete-deg-1} appears challenging.
Instead, we  handle the edges incident to unit-degree vertices \emph{algebraically}.
To allow for this, we first make a simplifying assumption about the graph:
specifically, we will argue that we can assume that the vertices deleted during the exhaustive removal of unit-degree vertices form stars rather than trees.

Let $\Gnew$ be the graph obtained by starting with $G$, and then exhaustively deleting all unit-degree vertices along with their incident edges.
If $\Gnew$ is empty, then $G$ must have been a tree.
In this case, we can solve the \listedgecolor{} problem in polynomial time on $G$. 

\begin{proposition}
\label{prop:tree}
    \listedgecolor{} can be solved in polynomial time over trees.
\end{proposition}

It is observed in \cite[Section 1]{DBLP:journals/ipl/Marx04} that
\Cref{prop:tree} follows from \cite{DBLP:journals/jgt/MarcotteS90}.

Suppose $G$ is not a tree.
Then $\Gnew$ is nonempty, and thus has minimum degree two.

Let $V_1$ denote the set of vertices deleted from $G$.
Then, $G[V_1]$ must be a forest.  
Let $\mathcal{T}$ be the collection of connected components of $G[V_1]$.
For each vertex $v$ in $\Gnew$, let $\mathcal{T}_v$ be the set of trees $T_v \in \mathcal{T}$ such that $v$ is adjacent to a vertex of $T$ by an edge in $G$. 
Since $G$ is connected, every tree in $\cal{T}$ belongs to $\mathcal{T}_v$ for some vertex $v$ in $\Gnew$.
The next result shows that in fact each tree in $\cal{T}$ belongs to a \emph{unique} set $\mathcal{T}_v$.

\begin{lemma}[Unique Connection]
\label{lem:unique-connection}
    For distinct vertices $v\neq w$ in $\Gnew$, we have $\mathcal{T}_v\cap\mathcal{T}_w=\emptyset$.
\end{lemma}
\begin{proof}
    Suppose to the contrary that there exists a tree $T\in \mathcal{T}_v\cap\mathcal{T}_w$.
    Let $e_v$ and $e_w$ be the edges connecting $T$ to vertices $v$ and $w$ respectively.
    Then there is a simple path $P$ from $v$ to $w$ in $G$,
    which begins with $e_v$, then uses edges in $T$, and ends in $w$. 
    Since $T$ only contains vertices from $V_1$, 
    every vertex in $T$ was deleted at some point for having unit degree.
    Let $u$ be the first vertex in $P$ that was deleted in this way.
    Then at the moment $u$ would be deleted, the path $P$ must still exist.
    But $u\in V_1$ is an internal node of $P$, and so has degree at least two in any graph containing $P$.
    This contradicts the fact that $u$ would be deleted for having unit degree.
    So our initial assumption was false, and $\mathcal{T}_v\cap\mathcal{T}_w=\emptyset$ as claimed. 
\end{proof}

Fix a vertex $v$ and tree $T_v \in \mathcal{T}_v$.
Suppose that $T_v$ is connected to $v$ by an edge $e = \{ v, w \} \in E$. 
Let $(T_v + e)$ denote the graph formed by adding edge $e$ to $T_v$.
Since $w$ was deleted at a time when it had degree one in the initial graph, it must be the case that $v$ has degree one in $(T_v+e)$, and thus $(T_v + e)$ is a tree with leaf $v$. 

For each color $i \in L_{e}$, we solve \textsc{List Edge Coloring}  on the tree $(T_v+e)$ where $L_e$ is replaced with $\{ i \}$ but all other edges in $T_v$ retain their lists from the input. 
By \Cref{prop:tree}, this takes polynomial time. 
After solving these instances,
we can identify the set $L_e'$ of all colors $i \in [k]$ such that $(T_v + e)$ admits a proper edge coloring with respect to the input lists in which $e$ receives color $i$.

\begin{lemma}[Pruning Trees]
\label{lem:prune}
Let $(G-T_v+e)$ be the graph obtained by deleting $T_v$ from $G$, and then adding back in vertex $w$ and edge $e = \set{v,w}$.
Then $G$ admits a proper edge coloring with respect to the lists $L_e$ if and only if $(G-T_v+e)$ admits a proper list edge coloring where edge $e$ has list $L_e'$, and all other edges retain their original lists. 
\end{lemma}
\begin{proof}
    Suppose $G$ has a proper edge coloring for the lists $L_e$.
    This coloring restricts to a proper edge coloring of the subgraph  $(T_v+e)$ with the same lists.
    Then by definition of the colors in $L_e'$, edge $e$ must be assigned a color from $L_e'$ in this coloring.
    So the original coloring restricts to a proper list edge coloring of the subgraph $(G-T_v+e)$ where the list $L_e$ for $e$ is replaced with $L_e'$, as claimed. 

    Conversely, suppose that $(G-T_v+e)$ has a proper list edge coloring where edge $e$ has list $L_e'$, and all other edges retain their lists from the input. Let $i\in L_e'$ be the color assigned to edge $e$.
    By definition of $L'_e$, 
    we know that $(T_v+e)$ admits a proper list edge coloring with the input lists where $e$ is assigned color $i$. 
    We combine this coloring with the previous coloring to obtain a list edge coloring of $G$ (both colorings assign edge $e$ the same color, and  the graphs have disjoint edge sets otherwise, so this is well-defined).     
    Since $e$ is the unique edge connecting $T_v$ to vertex $v$, and since by \Cref{lem:unique-connection} the graph $G$ does not contain edges connecting $T_v$ to any vertex besides $v$ in $G'$, 
    this combination yields a proper list edge coloring for $G$ with the desired properties.
\end{proof}

By iterating over all vertices $v$ in $\Gnew$, all trees $T_v\in\mathcal{T}_v$, and all colors in $[k]$,
the discussion in the paragraph before \Cref{lem:prune}
shows that we can determine the list of admissible colors $L_e'\sub L_e$ for each edge $e$ connecting some tree in $\mathcal{T}$ to a vertex in $\Gnew$ using \Cref{prop:tree} in polynomial time. 
If for some edge $e$ we have $L_e'=\emptyset$,
we get by definition of $L_e'$ that the instance does not admit a proper list edge coloring. 

So suppose that $L'_e\neq\emptyset$ for all edges $e$ connecting trees in $\mathcal{T}$ to vertices in $\Gnew$.
By repeated application of \Cref{lem:prune}, we can delete each tree $T\in\mathcal{T}$ from $G$, add back in the unique edge $e$ which connected $T$ to $\Gnew$, and replace the list of $L_e$ of $e$ with $L_e'$,
without changing whether the graph admits a proper list edge coloring.
This transformation ultimately replaces each tree in $\mathcal{T}$ with a single edge exiting a vertex in $\Gnew$.

Using the procedure from the previous paragraph, we get that without loss of generality, we may assume that $V_1$ consists solely of unit-degree vertices in $G$ (so that $\mathcal{T}$ consists of isolated nodes). 
For each vertex $v$ in $\Gnew$, we let $G_v$ denote the star graph formed by $v$ and its unit-degree neighbors in $V_1$.
We let $\Vnew$ and $\Enew$ be the vertex and edge sets of $\Gnew$, and
$\nnew$ and $\mnew$ denote the number of vertices and edges in $\Gnew$.
We can check in polynomial time that $G_v$ has a proper list edge coloring for all vertices $v$.
If $G_v$ does not have a proper list edge coloring for some $v$, we can immedately return NO in the \listedgecolor{} problem.
So in the sequel, we assume each $G_v$ admits a proper list edge coloring.

\subsection{Instantiating the Template}

We apply the matroid, matrix, and polynomial constructions from \Cref{sec:template}
to the graph $\Gnew$ described in \Cref{subsec:list-degree}, instead of the original input graph $G$. 

We make use of the matroid guaranteed by the following lemma. 

\begin{restatable}[Extension Matroid]{lemma}{extend}
\label{lem:extension}
  For each vertex $v$ in $\Gnew$, there is a linear matroid with ground set $S_v$, whose bases  $B \subseteq S_v$ are the sets of size $\deg_{\Gnew}(v)$ such that  $G_v$ admits a proper edge coloring where each edge receives a color from $L_e \setminus \{ i \mid v_i \in B \}$.
  Moreover, we can construct a $\deg_{\Gnew}(v)\times k$ matrix representing this matroid in randomized polynomial time with one-sided error.
  We call this matroid the \emph{extension matroid} of $\Gnew$ at $v$.
\end{restatable}

\Cref{lem:extension} is proved in \Cref{sec:construction}.

For each vertex $v$,
we set $\matroid_v$ to be the extension matroid over $S_v$ 
defined in \Cref{lem:extension}.
By \Cref{eq:W-matroid}, $\watroid$ is the direct sum of the matroids. 
By \Cref{prop:matroid-sum,lem:extension}, we can construct the matrix $B$ representing $\watroid$ in polynomial time.
In this case, condition 3 of collection $\cal{C}$ states that for a tuple $(M_1, \dots, M_k)$ to belong to $\cal{C}$, it must be the case that 
for every vertex $v$ in $\Gnew$,
$G_v$ admits a proper list edge coloring for the input lists,
where no edge in $G_v$ is assigned any color $i\in [k]$ such that $M_i$ contains an edge with $v$ as an endpoint. 

Intuitively, this condition ensures we only enumerate over tuples of matchings $(M_1, \dots, M_k)$ with the property that the list edge coloring on $\Gnew$ induced by assigning the edges in $M_i$ color $i$ can be \emph{extended} to a proper list edge coloring of the full graph $G$.

\begin{lemma}[List Coloring by Polynomials]
\label{lem:listcolor-char}
    The graph $G$ admits a proper list edge coloring for  lists $L_e$ if and only if the polynomial $P(X,Y)$
    has a monomial divisible by $\prod_{e\in \Enew} x_e$.
\end{lemma}
\begin{proof}
Suppose $P(X,Y) = \Pf BAB^\top$ contains a monomial divisible by $\prod_{e\in \Enew} x_e$.
Since $P$ is homogeneous of degree $\mnew$ in the $X$ variables, 
by \Cref{lem:generating-function}
there exists a $k$-tuple of \emph{edge-disjoint} matchings $\vec{M} = (M_1, \dots, M_{k})\in \cal{C}$, such that every edge of $\Gnew$ appears in some matching $M_i$. 
Consequently, assigning the edges in $M_i$ color $i$ yields a proper list edge coloring $c$ of $\Gnew$ for the lists $L_e$.
Since $\vec{M}\in\mathcal{C}$,
by condition 3 from the definition of $\mathcal{C}$,
we know that for every vertex $v$ in $\Gnew$, the graph $G_v$ admits a proper list edge coloring using none of the colors assigned by $c$ to the edges incident to $v$ in $\Gnew$.
Then by combining coloring $c$ for $\Gnew$ with the aforementioned proper list edge colorings for $G_v$ over all vertices $v$, we obtain a proper list edge coloring of the full graph $G$, as desired. 

Conversely, suppose $G$ has a proper list edge coloring for lists $L_e$.
For each $i\in [k]$,
let $M_i$ be the set of edges in $\Gnew$ assigned color $i$. 
Since this is a proper edge coloring, each $M_i$ is a matching in $G$. 
Since every edge is assigned a unique color, each vertex $v$ in $\Gnew$ is incident to exactly $\deg_{\Gnew}(v)$ edges across the matchings $M_i$.
Moreover, each edge $e$ in $G_v$ must be assigned a color from $L_e$ which 
 is not used by any other edges incident to $v$.
In particular the color assigned to $e$ cannot be $i\in [k]$ if $M_i$ contains an edge incident to $v$.
This implies that the set of copies $v_i$ taken over all $i\in [k]$ such that $M_i$ has an edge incident to $v$ forms an independent set in $\matroid_v$.
Thus $(M_1, \dots, M_{k})\in\cal{C}$.
Then by \Cref{lem:generating-function}, $P$ contains a monomial divisible by $\prod_{e\in \Enew} x_e$ as claimed. 
\end{proof}

\begin{lemma}[Algorithm for  Sparse Graphs]
\label{prop:main}
    There is a randomized algorithm that solves {\normalfont\textsc{List Edge Coloring}} with high probability and one-sided error on graphs with $n_1$ unit-degree vertices in $O^*(2^{m - (n + n_1)/2})$ time and polynomial space.
\end{lemma}
\begin{proof}
    Let $G$ be the input graph with $n$ vertices, $m$ edges, and $n_1$ unit-degree nodes.
    Run the algorithm from \Cref{lem:basic-dom} on $\Gnew$ to obtain a dominating set $D$ of $\Gnew$ that has size at most $\nnew/2$.
    Let $V' = \Vnew\setminus D$.
    Since $|D|\le \nnew/2$ we have $|V'| \ge \nnew/2 = (n-n_1)/2$.

    By \Cref{lem:listcolor-char}, $G$ has a proper list edge coloring if and only if $P(X,Y)$ has a monomial divisible by $\prod_{e\in\Enew} x_e$.
    By \Cref{lem:template}, we can determine whether $P$ has such a monomial in $O^*(2^{\mnew - |V'|})$ time and polynomial space. 
    Since 
        \[\mnew = m - |V_1| = m - n_1\]
    we see that the algorithm solves \listedgecolor{} in time at most 
        \[O^*(2^{(m - n_1) - (n - n_1)/2}) \le O^*(2^{m-(n+n_1)/2})\]
    as claimed. 
\end{proof}

\mainwlist*
\begin{proof}
  Let $n_1$ be the number of nodes in $G$ with degree one. 

    If $n_1 \ge n/5$, then by \Cref{prop:main} we can solve \listedgecolor{} in 
        \[O^*(2^{m - (n+n_1)/2}) \le O^*(2^{m-3n/5})\]
    time and polynomial space. 

    Otherwise, $n_1\le n/5$.
    In this case, by \Cref{lem:dominating-set} we can find a minimum-size dominating set $D$ of $\Gnew$ in $O^*(2^{m-3n/5})$ time and polynomial space. 
    Note that here we are using the assumption, justified in the final paragraph of \Cref{subsec:list-degree}, that the vertices deleted from $G$ to produce $\Gnew$ are precisely the vertices that have degree one in $G$. 
    Since $\Gnew$ has minimum degree two,
    by \Cref{lem:dominating-set-2} the set $D$ has at most $2n/5$ vertices.
    Let $V' = V\setminus D$.

    By \Cref{lem:listcolor-char}, $G$ has a proper list edge coloring if and only if $P(X,Y)$ has a monomial divisible by $\prod_{e\in\Enew} x_e$.
    By \Cref{lem:template}, we can determine whether $P$ has such a monomial in $O^*(2^{\mnew - |V'|})$ time and polynomial space. 
    Since $\mnew \le m$, this algorithm runs in 
        \[O^*(2^{\mnew - |V'|}) \le O^*(2^{m-3n/5})\]
    time and polynomial space as desired. 
\end{proof}

\section{Conclusion}
\label{sec:conclusion}

In this paper, we presented the first algorithms for \textsc{Edge Coloring} and \textsc{List Edge Coloring} which run in faster than $O^*(2^m)$ time while using only polynomial space. 
Our algorithm is based off a \emph{partition sieving} procedure that works over polynomials of degree $d$ and
partition matroids with $p$ parts.
We showed how to implement this sieving routine in
$O^*(2^{d - p})$ time when the polynomial is in some sense compatible with the matroid, 
an improvement over the $O^*(2^d)$ runtime that arose in previous techniques.

Overall, we solve \edgecolor{} in faster than $O^*(2^m)$ time by 
\begin{description}
    \item[Step 1] 
        designing a polynomial $P$ enumerating matchings meeting certain degree constraints, 
    \item[Step 2]
        finding a dominating set $D$ of size at most $2n/5$ in the input graph, and
    \item[Step 3]
         applying partition sieving with a partition matroid on $p=|D|$ parts. 
\end{description}

Implementing \textsf{\textbf{\textcolor{lipicsGray}{step 2}}}
above was possible for \edgecolor{} because we could assume that the input graph had minimum degree two without loss of generality, and such graphs on $n\ge 8$ nodes always have dominating sets of size at most $2n/5$.
For \listedgecolor{} we cannot make this assumption,
but nonetheless implement a variant of \textsf{\textbf{\textcolor{lipicsGray}{step 2}}} by identifying a subgraph with minimum degree two and using extension matroids to handle unit-degree vertices for free. 
For all sufficiently large $n$, it is also known that graphs on $n$ nodes with minimum degree three have
dominating sets of size at most $3n/8$ \cite{DBLP:journals/cpc/Reed96}.
If degree-two vertices could somehow be ``freely removed'' from the input graph just like unit-degree vertices can, one could use this bound on the dominating set to solve coloring problems faster. 
More general connections between the minimum degree of a graph and the minimum size of its dominating set have also been studied \cite{DBLP:journals/dmgt/Henning22}.

Understanding the precise time complexity of \edgecolor{} and its variants remains an interesting open research direction.
Is \textsc{Edge Coloring} solvable in $O^*(1.9999^m)$ time, or at least in $O^*(2^{m - n})$ time?
Establishing any nontrivial lower bounds for \edgecolor{} would also be very interesting. 
The reductions from  \cite{Holyer81a} imply that  solving \edgecolor{} in graphs with maximum degree $\Delta=3$ requires $2^{\Omega(n)}$ time, assuming the \textsf{Exponential Time Hypothesis (ETH)}, a standard hypothesis from fine-grained complexity.
Does \textsf{ETH} similarly imply a $2^{\Omega(m)}$ or even a $2^{\Omega(n\log n)}$ time lower bound for \edgecolor{} in dense graphs (this question was previously raised in \cite[Open problem 4.25]{DBLP:journals/dagstuhl-reports/LewensteinPW16})?

\bibliographystyle{plain}
\bibliography{main}

\newcommand{\noopsort}[1]{}  \newcommand{\noop}[1]{}
\begin{thebibliography}{10}

\bibitem{arnautov1974estimation}
Vladimir~I Arnautov.
\newblock Estimation of the exterior stability number of a graph by means of the minimal degree of the vertices.
\newblock {\em Prikl. Mat. i Programmirovanie}, 11(3-8):126, 1974.

\bibitem{near-linear-vizing}
Sepehr Assadi, Soheil Behnezhad, Sayan Bhattacharya, Martín Costa, Shay Solomon, and Tianyi Zhang.
\newblock Vizing's theorem in near-linear time, 2024.

\bibitem{BaBrKuOl2022}
Alkida Balliu, Sebastian Brandt, Fabian Kuhn, and Dennis Olivetti.
\newblock Distributed edge coloring in time polylogarithmic in {$\Delta$}.
\newblock In {\em Proceedings of the 2022 ACM Symposium on Principles of Distributed Computing}, PODC ’22, page 15–25. ACM, July 2022.

\bibitem{beineke1968}
Lowell~W Beineke.
\newblock Derived graphs and digraphs.
\newblock {\em Beitr{\"a}ge zur graphentheorie}, pages 17--33, 1968.

\bibitem{Bernshteyn2022}
Anton Bernshteyn.
\newblock A fast distributed algorithm for {$(\Delta+1)$}-edge-coloring.
\newblock {\em Journal of Combinatorial Theory, Series B}, 152:319–352, January 2022.

\bibitem{BhCoMaPaSo2024}
Sayan Bhattacharya, Mart{\'\i}n Costa, Nadav Panski, and Shay Solomon.
\newblock {Arboricity-Dependent Algorithms for Edge Coloring}.
\newblock In {\em Proceedings of the 19th Scandinavian Symposium and Workshops on Algorithm Theory (SWAT 2024)}, volume 294 of {\em Leibniz International Proceedings in Informatics (LIPIcs)}, pages 12:1--12:15, Dagstuhl, Germany, 2024. Schloss Dagstuhl -- Leibniz-Zentrum f{\"u}r Informatik.

\bibitem{Bjorklund14detsum}
Andreas Bj{\"{o}}rklund.
\newblock Determinant sums for undirected {H}amiltonicity.
\newblock {\em {SIAM} Journal on Computing}, 43(1):280--299, 2014.

\bibitem{DBLP:journals/corr/abs-2404-04987}
Andreas Bj{\"{o}}rklund, Radu Curticapean, Thore Husfeldt, Petteri Kaski, and Kevin Pratt.
\newblock Chromatic number in $1.9999^n$ time? fast deterministic set partitioning under the asymptotic rank conjecture.
\newblock {\em CoRR}, abs/2404.04987, 2024.

\bibitem{BjorkHusf2007}
Andreas Bj\"{o}rklund and Thore Husfeldt.
\newblock Exact algorithms for exact satisfiability and number of perfect matchings.
\newblock {\em Algorithmica}, 52(2):226–249, December 2007.

\bibitem{BjorklundHusfeldtKaskiKoivisto2009}
Andreas Bj\"{o}rklund, Thore Husfeldt, Petteri Kaski, and Mikko Koivisto.
\newblock Trimmed moebius inversion and graphs of bounded degree.
\newblock {\em Theory of Computing Systems}, 47(3):637–654, January 2009.

\bibitem{BjorklundHKK17narrow}
Andreas Bj{\"{o}}rklund, Thore Husfeldt, Petteri Kaski, and Mikko Koivisto.
\newblock Narrow sieves for parameterized paths and packings.
\newblock {\em Journal of Computer and System Sciences}, 87:119--139, 2017.

\bibitem{BjHuKo09}
Andreas Bj{\"{o}}rklund, Thore Husfeldt, and Mikko Koivisto.
\newblock Set partitioning via inclusion-exclusion.
\newblock {\em {SIAM} Journal on Computing}, 39(2):546--563, 2009.

\bibitem{blank1973estimate}
M~Blank.
\newblock An estimate of the external stability number of a graph without suspended vertices.
\newblock {\em Prikl. Mat. i Programmirovanie}, 10:3--11, 1973.

\bibitem{Byskov2004}
Jesper~Makholm Byskov.
\newblock Enumerating maximal independent sets with applications to graph colouring.
\newblock {\em Operations Research Letters}, 32(6):547–556, November 2004.

\bibitem{chrRoVlie2024}
Aleksander B.~G. Christiansen, Eva Rotenberg, and Juliette Vlieghe.
\newblock {Sparsity-Parameterised Dynamic Edge Colouring}.
\newblock In {\em Proceedings of the 19th Scandinavian Symposium and Workshops on Algorithm Theory (SWAT 2024)}, volume 294 of {\em Leibniz International Proceedings in Informatics (LIPIcs)}, pages 20:1--20:18, Dagstuhl, Germany, 2024. Schloss Dagstuhl -- Leibniz-Zentrum f{\"u}r Informatik.

\bibitem{DBLP:journals/mst/CouturierGKLP13}
Jean{-}Fran{\c{c}}ois Couturier, Petr~A. Golovach, Dieter Kratsch, Mathieu Liedloff, and Artem~V. Pyatkin.
\newblock Colorings with few colors: Counting, enumeration and combinatorial bounds.
\newblock {\em Theory Comput. Syst.}, 52(4):645--667, 2013.

\bibitem{Duan2019}
Ran Duan, Haoqing He, and Tianyi Zhang.
\newblock {\em Dynamic Edge Coloring with Improved Approximation}, page 1937–1945.
\newblock Society for Industrial and Applied Mathematics, January 2019.

\bibitem{determinantalsieving}
Eduard Eiben, Tomohiro Koana, and Magnus Wahlstr{\"{o}}m.
\newblock Determinantal sieving.
\newblock In David~P. Woodruff, editor, {\em Proceedings of the 2024 {ACM-SIAM} Symposium on Discrete Algorithms ({SODA} 2024)}, pages 377--423. {SIAM}, 2024.

\bibitem{Eppstein2001}
David Eppstein.
\newblock {\em Small Maximal Independent Sets and Faster Exact Graph Coloring}, page 462–470.
\newblock Springer Berlin Heidelberg, 2001.

\bibitem{DBLP:journals/algorithmica/GaspersL23}
Serge Gaspers and Edward~J. Lee.
\newblock Faster graph coloring in polynomial space.
\newblock {\em Algorithmica}, 85(2):584--609, 2023.

\bibitem{DBLP:journals/talg/GolovnevKM16}
Alexander Golovnev, Alexander~S. Kulikov, and Ivan Mihajlin.
\newblock Families with infants: Speeding up algorithms for {NP}-hard problems using {FFT}.
\newblock {\em {ACM} Trans. Algorithms}, 12(3):35:1--35:17, 2016.

\bibitem{Harris2019}
David~G. Harris.
\newblock Distributed local approximation algorithms for maximum matching in graphs and hypergraphs.
\newblock In {\em Proceedings of the 60th Annual Symposium on Foundations of Computer Science (FOCS 2019)}, pages 700--724, 2019.

\bibitem{DBLP:journals/dmgt/Henning22}
Michael~A. Henning.
\newblock Bounds on domination parameters in graphs: a brief survey.
\newblock {\em Discuss. Math. Graph Theory}, 42(3):665--708, 2022.

\bibitem{Holyer81a}
Ian Holyer.
\newblock The {NP}-completeness of edge-coloring.
\newblock {\em {SIAM} Journal on Computing}, 10(4):718--720, 1981.

\bibitem{IshikawaW95}
Masao Ishikawa and Masato Wakayama.
\newblock Minor summation formula of pfaffians.
\newblock {\em Linear and Multilinear algebra}, 39(3):285--305, 1995.

\bibitem{koana2024faster}
Tomohiro Koana and Magnus Wahlstr{\"o}m.
\newblock Faster algorithms on linear delta-matroids.
\newblock {\em arXiv preprint arXiv:2402.11596}, 2024.

\bibitem{DBLP:journals/tcs/Kowalik09}
Lukasz Kowalik.
\newblock Improved edge-coloring with three colors.
\newblock {\em Theor. Comput. Sci.}, 410(38-40):3733--3742, 2009.

\bibitem{KowalikSocala2018}
Lukasz Kowalik and Arkadiusz Socala.
\newblock {Tight Lower Bounds for List Edge Coloring}.
\newblock In David Eppstein, editor, {\em 16th Scandinavian Symposium and Workshops on Algorithm Theory (SWAT 2018)}, volume 101 of {\em Leibniz International Proceedings in Informatics (LIPIcs)}, pages 28:1--28:12, Dagstuhl, Germany, 2018. Schloss Dagstuhl -- Leibniz-Zentrum f{\"u}r Informatik.

\bibitem{KM2024}
Alexander~S Kulikov and Ivan Mihajlin.
\newblock If edge coloring is hard under {SETH}, then {SETH} is false.
\newblock In {\em 2024 Symposium on Simplicity in Algorithms ({SOSA})}, pages 115--120. Society for Industrial and Applied Mathematics, Philadelphia, PA, January 2024.

\bibitem{Lawler76}
Eugene~L. Lawler.
\newblock A note on the complexity of the chromatic number problem.
\newblock {\em Inf. Process. Lett.}, 5(3):66--67, 1976.

\bibitem{DBLP:journals/dagstuhl-reports/LewensteinPW16}
Moshe Lewenstein, Seth Pettie, and Virginia {Vassilevska Williams}.
\newblock Structure and hardness in {P} (dagstuhl seminar 16451).
\newblock {\em Dagstuhl Reports}, 6(11):1--34, 2016.

\bibitem{LokshtanovMPS18TALG}
Daniel Lokshtanov, Pranabendu Misra, Fahad Panolan, and Saket Saurabh.
\newblock Deterministic truncation of linear matroids.
\newblock {\em {ACM} Transactions on Algorithms}, 14(2):14:1--14:20, 2018.

\bibitem{Lovasz1973}
L{\'a}szl{\'o} Lov{\'a}sz.
\newblock Coverings and colorings of hypergraphs.
\newblock {\em Proc. 4th Southeastern Conference of Combinatorics, Graph Theory, and Computing}, pages 3--12, 1973.

\bibitem{LundYannakis1994}
Carsten Lund and Mihalis Yannakakis.
\newblock On the hardness of approximating minimization problems.
\newblock {\em Journal of the ACM}, 41(5):960–981, September 1994.

\bibitem{DBLP:journals/jgt/MarcotteS90}
Odile Marcotte and Paul~D. Seymour.
\newblock Extending an edge-coloring.
\newblock {\em J. Graph Theory}, 14(5):565--573, 1990.

\bibitem{DBLP:journals/ipl/Marx04}
D{\'{a}}niel Marx.
\newblock List edge multicoloring in graphs with few cycles.
\newblock {\em Inf. Process. Lett.}, 89(2):85--90, 2004.

\bibitem{Marx09-matroid}
D{\'a}niel Marx.
\newblock A parameterized view on matroid optimization problems.
\newblock {\em Theoretical Computer Science}, 410(44):4471--4479, 2009.

\bibitem{mccuaig1989domination}
William McCuaig and Bruce Shepherd.
\newblock Domination in graphs with minimum degree two.
\newblock {\em Journal of Graph Theory}, 13(6):749--762, 1989.

\bibitem{Meijer2023}
Lucas Meijer.
\newblock 3-coloring in time {$O(1.3217^n)$}, 2023.

\bibitem{Motwani1995}
Rajeev Motwani and Prabhakar Raghavan.
\newblock {\em Randomized Algorithms}.
\newblock Cambridge University Press, August 1995.

\bibitem{murota1999matrices}
Kazuo Murota.
\newblock {\em Matrices and matroids for systems analysis}, volume~20.
\newblock Springer Science \& Business Media, 1999.

\bibitem{ore1962theory}
Oystein Ore.
\newblock Theory of graphs.
\newblock In {\em Colloquium Publications}. American Mathematical Society, 1962.

\bibitem{OxleyBook2}
James Oxley.
\newblock {\em Matroid Theory}.
\newblock Oxford University Press, 2011.

\bibitem{payan1975nombre}
Charles Payan.
\newblock Sur le nombre d'absorption d'un graphe simple.
\newblock 1975.

\bibitem{DBLP:journals/cpc/Reed96}
Bruce~A. Reed.
\newblock Paths, stars and the number three.
\newblock {\em Comb. Probab. Comput.}, 5:277--295, 1996.

\bibitem{DBLP:journals/dam/RooijB11}
Johan M.~M. van Rooij and Hans~L. Bodlaender.
\newblock Exact algorithms for dominating set.
\newblock {\em Discret. Appl. Math.}, 159(17):2147--2164, 2011.

\bibitem{Wahlstrom13STACS}
Magnus Wahlstr{\"{o}}m.
\newblock Abusing the {T}utte matrix: An algebraic instance compression for the {K}-set-cycle problem.
\newblock In {\em Proceedings of the 30th International Symposium on Theoretical Aspects of Computer Science ({STACS} 2013)}, volume~20 of {\em LIPIcs}, pages 341--352. Schloss Dagstuhl - Leibniz-Zentrum f{\"{u}}r Informatik, 2013.

\bibitem{Williams09IPL}
Ryan Williams.
\newblock Finding paths of length $k$ in ${O^*(2^k)}$ time.
\newblock {\em Information Processing Letter}, 109(6):315--318, 2009.

\bibitem{WuGuJiangShaoXu2024}
Pu~Wu, Huanyu Gu, Huiqin Jiang, Zehui Shao, and Jin Xu.
\newblock {A Faster Algorithm for the 4-Coloring Problem}.
\newblock In Timothy Chan, Johannes Fischer, John Iacono, and Grzegorz Herman, editors, {\em 32nd Annual European Symposium on Algorithms (ESA 2024)}, volume 308 of {\em Leibniz International Proceedings in Informatics (LIPIcs)}, pages 103:1--103:18, Dagstuhl, Germany, 2024. Schloss Dagstuhl -- Leibniz-Zentrum f{\"u}r Informatik.

\bibitem{DBLP:conf/icalp/Zamir21}
Or~Zamir.
\newblock Breaking the $2^n$ barrier for 5-coloring and 6-coloring.
\newblock In {\em Proceedings of the 48th International Colloquium on Automata, Languages, and Programming ({ICALP} 2021)}, volume 198 of {\em LIPIcs}, pages 113:1--113:20. Schloss Dagstuhl - Leibniz-Zentrum f{\"{u}}r Informatik, 2021.

\bibitem{DBLP:conf/stoc/Zamir23}
Or~Zamir.
\newblock Algorithmic applications of hypergraph and partition containers.
\newblock In {\em Proceedings of the 55th Annual {ACM} Symposium on Theory of Computing ({STOC} 2023)}, pages 985--998. {ACM}, 2023.

\end{thebibliography}

\appendix

\section{Additional Related Work}
\label{sec:related}

\paragraph*{Coloring With Few Colors}
Beyond the general coloring problems discussed previously,
there is a large body of work on algorithms for the parameterized version of \vertexcolor{}, the $k$-\coloring{} problem.
In this problem, we are given an undirected graph $G$ on $n$ vertices and a positive integer $k$, and are tasked with determining if $G$ contains a proper vertex coloring which uses at most $k$ colors.
For $k\le 2$ this problem can be solved in polynomial time, but for $k\ge 3$ this problem is \NP-hard \cite{Lovasz1973}.

Since the current fastest algorithm for \vertexcolor{} takes $O^*(2^n)$ time, 
much attention has been spent designing routines which can solve $k$-\coloring{} in faster than $O^*(2^n)$ time for small values of $k$.
It is known that 3-\coloring{} can be solved in $O^*(1.3217^n)$ time \cite{Meijer2023},
4-\coloring{} can be solved in $O^*(1.7159^n)$ time \cite{WuGuJiangShaoXu2024}, 
and for $k\in\set{5,6}$ the $k$-\coloring{} problem can be solved in $O^*(2^{(1-\eps)n})$ time for an extremely small constant $\eps > 0$ \cite{DBLP:conf/icalp/Zamir21}.
For $k\ge 7$, it remains open whether $k$-\coloring{} can be solved in exponentially faster than the current $O^*(2^n)$ time upper bound  for \vertexcolor{}. 

Similarly, it is known that \textsc{Edge Coloring} can be solved in faster than $O^*(2^m)$ time for graphs with small maximum degree.
For example,  \textsc{Edge Coloring} can be solved in $O^*(1.2179^m)$ time and polynomial space when $\Delta = 3$ \cite{DBLP:journals/tcs/Kowalik09}, and 
for each positive integer $\Delta\le 6$ there exists a constant $\eps_{\Delta} > 0$ such that \edgecolor{} can be solved over graphs of maximum degree $\Delta$ in $O^*(2^{(1 - \varepsilon)m})$ time  and exponential space \cite{DBLP:journals/mst/CouturierGKLP13}. 

\paragraph*{Graph Classes}

There has also been significant research on solving coloring problems faster on structured families of graphs. 
As mentioned earlier, \vertexcolor{} can be solved over graphs with maximum degree $\Delta$ in $O^*(2^{(1-\eps_\Delta)n})$ time for $\eps_\Delta = (1/2)^{\Theta(\Delta)}$ \cite{BjorklundHusfeldtKaskiKoivisto2009} and
over graphs of average degree $d$ in $O^*(2^{(1-\eps_d)n})$ time for $\eps_d = (1/d)^{\Theta(d^3)}$ \cite{DBLP:journals/talg/GolovnevKM16}, while
\edgecolor{} can be solved over $d$-regular graphs  in $O^*(2^{(1-(1/d))m})$ time \cite[Theorem 6]{BjorklundHKK17narrow}.
It is also known that in graphs where the maximum degree $\Delta \le cd$ is bounded above by a constant multiple $c$ of the average degree $d$, \vertexcolor{} can be solved in $O^*(2^{(1-\eps_c)n})$ time where $\eps_c > 0$ is a constant depending only on $c$, which rapidly approaches zero as $c$ increases \cite{DBLP:conf/stoc/Zamir23}. 
Finally, assuming the asymptotic rank conjecture (a hypothesis concerning the rank of certain tensors), \vertexcolor{} can be solved on general graphs in  $O^*(1.9999^n)$ time \cite{DBLP:journals/corr/abs-2404-04987}.
The algorithms discussed 
in this paragraph all require exponential space. 
Solving \vertexcolor{} in $O^*(2^n)$ time and polynomial space remains a major open question---the current fastest polynomial space algorithm for \vertexcolor{} takes $O^*(2.2356^n)$ time \cite{DBLP:journals/algorithmica/GaspersL23}.

\paragraph*{Lower Bounds}
Assuming the \textsf{Exponential Time Hypothesis (ETH)},
a popular conjecture in complexity, the generalization of \listedgecolor{} where the lists of colors can be arbitrary requires $2^{\Omega(n^2)}$ time to solve in general \cite{KowalikSocala2018}. 
It is a major open problem to show a similar lower bound for the \edgecolor{} problem.
It is known that one cannot show a $2^{\Omega(n^2)}$ lower bound for \edgecolor{} or \listedgecolor{} by a reduction from the \textsf{Strong Exponential Time Hypothesis (SETH)}, another standard conjecture in complexity that is stronger than \textsf{ETH}, without in fact falsifying \textsf{SETH}---this gives a barrier to some techniques for showing hardness of these problems \cite{KM2024}. 

\section{Additional Matroid and Pfaffian Constructions}
\label{sec:construction}

We first prove \Cref{lemma:direct-sum}, which asserts that Pfaffians play well with direct sums.

\pfaffiansum*
\begin{proof}
    Let $V_1$ and $V_2$ be disjoint sets indexing the rows and columns of $A_1$ and $A_2$ respectively.

    Let
        \[A = \begin{pmatrix}
       A_1 & O \\ O & A_2 
    \end{pmatrix}\]
be the matrix from the lemma statement.
    
By construction, the set $V_1\sqcup V_2$ indexes the rows and columns of $A$.

    Then expanding out the definition of the Pfaffian from \cref{eq:pfaffian-def}, we get that 
        \[\Pf A = \sum_{M\in \Pi(V_1\cup V_2)} \prod_{\set{u,v}\in M} A[u,v].\]

    Note that $A[u,v]$ is nonzero if and only if either $u,v\in V_1$ or $u,v\in V_2$.
    Consequently, the term corresponding to the matching
    $M\in \Pi(V_1\cup V_2)$ 
    from the sum in the right-hand side above is nonzero if and only if $M$ can be written as the disjoint union $M = M_1\sqcup M_2$ of matchings $M_1 \in \Pi(V_1)$ and $M_2\in \Pi(V_2)$.

    Thus, the above equation simplifies to 
    \[\Pf A = \sum_{\substack{M_1 = \Pi(V_1) \\ M_2 = \Pi(V_2)}} \grp{\prod_{\set{u,v}\in M_1} A[u,v] \cdot \prod_{\set{u,v}\in M_2} A[u,v]}.\]

    Given $u,v\in V_1$, by definition we have $A[u,v] = A_1[u,v]$.
    Similarly, given $u,v\in V_2$, we have $A[u,v] = A_2[u,v]$.
    Substituting this into the above equation we have 
            \[\Pf A = \sum_{\substack{M_1 = \Pi(V_1) \\ M_2 = \Pi(V_2)}} \grp{\prod_{\set{u,v}\in M_1} A_1[u,v] \cdot \prod_{\set{u,v}\in M_2} A_2[u,v]}
            =
            (\Pf A_1)(\Pf A_2)\]
    where the final equality follows by multiplying out the definitions of $\Pf A_1$ and $\Pf A_2$ obtained from  \cref{eq:pfaffian-def}.
\end{proof}

We now introduce more matroid concepts and constructions,
not covered in \Cref{sec:prelim}.

Given a bipartite graph $H$ with vertex parts $X$ and $U$, 
we define the \emph{transversal matroid} on $X$ with respect to $H$ 
to be the matroid with ground set $X$, where a set $S\sub X$ is independent if and only if $H$ contains a matching $M$ where every vertex in $S$ is incident to an edge in $M$ (i.e., the matching saturates $S$).
It is well known that a transversal matroid is indeed a matroid \cite[Section 1.5]{OxleyBook2}.
The following result lets us represent these matroids.

\begin{proposition}[Transversal Matroid Representation {\cite[Proposition 3.11]{Marx09-matroid}}]
\label{prop:transversal}
    There is a randomized polynomial time algorithm which takes in a bipartite graph $H$ with vertex parts $X$ and $U$ and with high probability and one-sided error returns a matrix $A$ representing the transversal matroid on $X$ with respect to $H$.
\end{proposition}

Given a matroid $\matroid = (X,\cI)$,
we define the \emph{dual} matroid $\matroid^*$ by taking the ground set $X$ together with the collection of all $S\sub X$ with the property that $S\cap B=\emptyset$ is disjoint from some basis $B$ of $\matroid$.
The dual of a matroid is always a matroid \cite[Section 3.3]{OxleyBook2}.

\begin{proposition}[Dual Representation]
\label{prop:dual}
    Given a matrix $A$ representing a matroid $\matroid$,
    we can compute a matrix representing the dual of $\matroid$ in polynomial time. 
\end{proposition}

See \cite[Proposition 3.6]{Marx09-matroid} for a proof of \Cref{prop:dual}.

We define one last matroid operation, called truncation.
Given a matroid $\matroid = (X,\cI)$ of rank $r$ and a nonnegative integer $h\le r$, the \emph{$h$-truncation} of $\matroid$ is defined to be the set $X$ together with the family of sets in $\cI$ of size at most $h$. 
It is well known that the $h$-truncation of a matroid is a matroid.

\begin{proposition}[Representing Truncations {\cite[Proposition 3.7]{Marx09-matroid}}]
\label{prop:truncate}
    Given a matrix $A$ representing a matroid $\matroid$ of rank $r$,
    and a nonnegative integer $h\le r$,
    we can in randomized polynomial time with high probability output a matrix with $h$ rows and the same column set as $A$ that represents the $h$-truncation of $\matroid$.
\end{proposition}

We now combine the above matroid operations to prove \Cref{lem:extension}.

\extend*
\begin{proof}
    Let $E_v$ be the set of edges in $G_v$.
    We define the bipartite graph $H$ with vertex parts $S_v$ and $E_v$,
    where $v_i\in S_v$ has an edge to $e\in E_v$ if and only if $i\in L_e$.
    Define $\matroid$ to be the transversal matroid on $S_v$ with respect to $H$. 
    By the assumption from the last paragraph of \Cref{subsec:list-degree}, $G_v$ admits a proper list edge coloring.
    Consequently, $\matroid$ has rank $|E_v|$.    
    
    Let $\matroid^*$ be the dual of $\matroid$.
    Set $h = \deg_{\Gnew}(v)$.
    Let $\matextend$ be the $h$-truncation of $\matroid^*$.
    By definition, the bases of $\matextend$ are precisely the sets $B\sub S_v$ of size $h$ such that $\matroid$ has a basis $U\sub S_v\setminus B^*$. 
    By the definition of $\matroid$, $H$ must contain a matching $M$ from $U$ to $E_v$.

    For each $v_i\in U$, assign the edge $e\in E_v$ connected to $v_i$ by $M$ the color $i$. 
    By definition of the graph $H$, we have $i\in L_e$.
    Moreover, $U\cap B=\emptyset$, so $i\not\in\set{j\mid v_j\in B}$.
    Finally, since $B$ is a basis of $\matroid$, $|B| = |E_v|$, so every edge in $E_v$ is assigned a color in this procedure. 
    Thus, $G_v$ admits a proper edge coloring where each edge receives a color from $L_e\setminus\set{i\mid v_i\in B}$.
    
    Similar reasoning shows that any set $B$ of size $h$ for which the previous sentence holds is a basis in $\matextend$.
    Thus, $\matextend$ is the extension matroid of $\Gnew$ at $v$ described in the lemma statement.
    We can construct an $h\times k$ matrix representing $\matextend$ in randomized polynomial time with one-sided error by combining \Cref{prop:transversal,prop:dual,prop:truncate}.
\end{proof}

\section{Dominating Set Algorithm}
\label{sec:dom}

\simpledomset*
\begin{proof}
    Let $D$ be a minimal dominating set.
    By the minimality of $D$, every vertex in $D$ has at least one neighbor in $V \setminus D$.
    So, $V \setminus D$ is a dominating set as well.
    Since $\min \{|D|, |V \setminus D| \} \le n/2$, the graph has a dominating set of size at most $n/2$ as claimed.
    
    We can construct a minimal dominating set by initializing $D\leftarrow V$, and repeatedly trying to delete vertices from $D$ while preserving the property that $D$ is a dominating set.
    Since we only try deleting each vertex once, and we can check if a given set is dominating in linear time, this procedure takes polynomial time overall. 
    Once we obtain a minimal dominating set $D$, by the discussion in the previous paragraph, we simply return the smaller of $D$ and $V\setminus D$ as our dominating set of size at most $n/2$.
\end{proof}

\domsetalgorithm*

\begin{proof}
We first observe some simple properties of dominating sets in \emph{paths}.

        \begin{claim} \label{obs:ds-path}
            There is a polynomial time algorithm that takes in a path $P = (v_1, \dots, v_p)$ on $p$ nodes and returns a minimum dominating set $D$ of $P$ of size exactly $\lceil p/3 \rceil$, such that 
            \begin{enumerate}
                \item if $p\equiv 1\pmod 3$ then $D$ includes both endpoints of $P$,
                \item if  $p\equiv 2\pmod 3$ then $D$ includes exactly one endpoint of $P$,
                and
                \item if $p\equiv 0\pmod 3$ then $D$ includes neither endpoint of $P$.
            \end{enumerate}
            Moreover, if $p\equiv 0\pmod 3$ then a minimum dominating set cannot include either endpoint of $P$, and if $p\equiv 2\pmod 3$ it cannot include both endpoints of $P$.  
        \end{claim}
        \begin{claimproof}
            Since each vertex in $P$ dominates at most 3 nodes in the path, the minimum dominating set has size at least $\lceil p/3 \rceil$.

        If $p\equiv 1\pmod 3$, then $\set{v_1, v_4, \dots, v_{p}}$ forms a dominating set of this size for $P$, which we return. 

        If $p\equiv 2\pmod 3$, then both $\set{v_2, v_5, \dots, v_{p}}$ and $\set{v_1, v_4, \dots, v_{p-1}}$ are dominating sets of the relevant size for $P$, and we can return either one. 
            
            If $p\equiv 0\pmod 3$, then $\set{v_2, v_5, \dots, v_{p - 1}}$ forms a dominating set of this size for $P$, which we return. 

        Let $D$ be a dominating set of $P$ that contains an endpoint of the path.
        If $p\equiv 0\pmod 3$,
        then the endpoint in $D$ dominates two vertices and the rest of the nodes in $D$ dominate the remaining vertices in $P$, so $D$ has size at least $1 + \lceil (p-2)/3\rceil > \lceil p/3\rceil$,
        so $D$ cannot be a minimum-size dominating set in this case.

        Now suppose $D$ is a dominating set of $P$ containing both endpoints of the path.
        If $p\equiv 2\pmod 3$,
        then the endpoints in $D$ dominate four vertices and the rest of the nodes in $D$ dominate the remaining vertices in $P$, so $D$ has size at least $2 + \lceil (p-4)/3\rceil > \lceil p/3\rceil$,
        so $D$ cannot be a minimum-size dominating set in this case either.

        This completes the proof.
        \end{claimproof}

    Let $V_1$ be the set of vertices deleted from $G$ when constructing $G'$.
    Furthermore, let $V_2$ and $V_{3}$ denote the set of degree-two vertices and vertices of degree at least three in $G'$ respectively.
    Define $n_1 = |V_1|$, $n_2 = |V_2|$, and $n_3 = |V_3|$.
    We  show how to find a minimum-size dominating set in $G'$ with runtime depending only on $n_3$.

    \begin{claim}[Searching Over High Degree Nodes]
    \label{claim:forsure}
        We can find a minimum  size dominating set of $G'$  in $O^*(2^{n_3})$ time and polynomial space. 
    \end{claim}
    \begin{claimproof}
        Let $D$ be a minimum-size dominating set that maximizes the value of $|D \cap V_3|$.
        Our goal is to find $D$. 
        We do this by trying out all possibilities for $D_3 = D \cap V_3$ in $O(2^{n_3})$ time.
        In what follows, fix a guess for the set $D_3$.
        To finish constructing the dominating set, it remains to determine which additional vertices in $V_2$ we will include. 

        Since every vertex in $V_2$ has degree two by definition in $G'$,
        every vertex in $G'[V_2]$ has degree at most two. 
        Thus every connected components of $G'[V_2]$ must be a path or cycle.

        Since $G$ is connected, the only way a component of $G'[V_2]$ can be a cycle is if in fact this component is the full graph $G' = G'[V_2]$.
        In this case, if the cycle has $p$ vertices, then its minimum dominating set has size at least $\lceil p/3\rceil$ because each vertex in the cycle dominates exactly three nodes, 
        and by \Cref{obs:ds-path} we can find a minimum dominating set of this size in polynomial time.

        Otherwise, the connected components of $G'[V_2]$ are all paths. 
        For each connected component in this graph, we must decide which of its nodes to include in the dominating set we construct. 
        There are three types of paths that may appear as connected components in $G'[V_2]$, based on how their endpoints are adjacent to vertices in $D_3$.

        \begin{description}
        \item[Type 1]
            Both endpoints are adjacent to nodes in $D_3$.
        \item[Type 2]
            One endpoint is adjacent to a node in $D_3$ and the other is not.
        \item[Type 3]
            Neither endpoint is adjacent to a node in $D_3$.
        \end{description}

        Below, we perform casework on these three different types of paths,
        and describe which nodes in these paths we include in the dominating set. 
        Note that in all cases, because $P$ consists of nodes in $G'[V_2]$,
        the endpoints of $P$ are adjacent to unique vertices in $V_3$, while the internal nodes of $P$ are not adjacent to any vertices in $V_3$. 
        
        We assume that we guessed the subset $D_3 = D\cap V_3$ correctly,
        and construct a minimum-size dominating set with this assumption in mind.
        After trying out all possible $D_3$ sets, at the end we will return the minimum-size dominating set that arose from each of these cases.
        In what follows, we let $P = (v_1, \dots, v_p)$ denote the path we are considering (so that $P$ has length $p$, and endpoints $v_1$ and $v_p$). 
        We say a vertex $v$ is \emph{dominated} if it is a node in or adjacent to a node in the dominating set we are building up. 

        \proofsubparagraph*{Type 1}
        In this case, because the endpoints of $P$
        are already dominated and the internal nodes of $P$ are not adjacent to nodes in $V_3$,
        we get a minimum dominating by applying the construction of \Cref{obs:ds-path} to $P$ with its endpoints removed.

        \proofsubparagraph{Type 2}
        Without loss of generality, suppose that $v_1$ is adjacent to a node in $D_3$,
        and that $v_p$ is not. 
        Let $w \in V_3\setminus D_3$ be the unique vertex that $v_p$ is adjacent to in $V_3$.
        Since $v_1$ is dominated by $D_3$, it remains to dominate the path $P'$ on the remaining $p' = p - 1$ vertices.
        We do this by casework on the length $p'$ of this remainder:
        \begin{itemize}

            \item If $p' \equiv 0 \pmod 3$, then we get a dominating set of size $p'/3$ vertices for $P'$ using \Cref{obs:ds-path}.
            \item If $p' \equiv 1 \pmod 3$, then
            by \Cref{obs:ds-path} the minimum dominating set within $P'$ has size $(p'+2)/3$.
            However, we could also dominate $P'$ with this many vertices by using $w$ (to dominate $v_p$) and then using \Cref{obs:ds-path} to dominate the remaining vertices of $P'$ using $(p'-1)/3$ additional vertices within the path. 
            This choice of vertices would still minimize the size of the dominating set, while including more vertices from $V_3$. 
            So in this case we know that our guess for $D_3 = D\cap V_3$ was incorrect, contradicting our assumption. 
            \item If $p' \equiv 2 \pmod 3$, we use \Cref{obs:ds-path} to dominate the path with $(p'+1)/3$ vertices.
            The construction from \Cref{obs:ds-path} includes $v_p$, and so $w$ is dominated by this set. 
        \end{itemize}

        \proofsubparagraph*{Type 3}
        Let $w_1$ and $w_p$ be the unique vertices in $V_3\setminus D_3$ adjacent to $v_1$ and $v_p$ respectively. 
        We  perform casework based on the length $p$ of the path $P$. 
        \begin{itemize}
            \item If $p \equiv 0 \pmod 3$, we use \Cref{obs:ds-path} to dominate the path with $p/3$ vertices.
            If we guess $D_3 = D\cap V_3$ correctly,
            then $w_1$ and $w_p$ must be dominated by vertices outside of $P$. 
            
            \item If $p \equiv 1 \pmod 3$, we use \Cref{obs:ds-path} to dominate the path with $(p + 2)/3$ vertices.
            The construction from \Cref{obs:ds-path} includes $v_1$ and $v_p$ in the dominating set, so in this case their neighbors $w_1$ and $w_p$ are dominated. 
            \item If $p \equiv 2 \pmod 3$, we can use \Cref{obs:ds-path} to dominate the path with $(p+1)/3$ vertices.
            Moreover, by \Cref{obs:ds-path} any such dominating set includes exactly one endpoint of $P$, and we can freely choose whether we decide to include $v_1$ or $v_p$.
            Put all paths $P$ that appear in this subcase into a family of paths $\mathcal{P}$.
            We discuss this case further below. 
        \end{itemize}

        So far we have described how to include vertices from paths in the minimum dominating set extending $D_3$, except in the case where the path $P$ falls into \textsf{\textbf{\textcolor{lipicsGray}{type 3}}} and has length congruent to 2 modulo 3. 
        We have collected all such paths $P$ into a family $\mathcal{P}$.

        We form a graph $H$ that has a node for each vertex in $V_3$ not dominated by any of the vertices we have included in the dominating set so far (i.e., not dominated by any nodes in $D_3$ or any nodes from the dominating sets for the path components chosen in the above case analysis),
        and that also has a node for each path $P\in\mathcal{P}$.
        We include an edge in $H$ between a path $P\in\mathcal{P}$ and a vertex $w\in V_3$
        if and only if $w$ is a neighbor of one of the endpoints of $P$. 

        This graph $H$ is bipartite between its nodes corresponding to paths $P\in\mathcal{P}$ and its nodes corresponding to vertices in $V_3$.
        Since each endpoint of $P$ is adjacent to a unique vertex in $V_3$ (because $P$ is of \textsf{\textbf{\textcolor{lipicsGray}{type 3}}}),
        each path $P$ viewed as a node in $H$ has degree at most two. 

        Assuming we guessed $D_3 = D\cap V_3$ correctly,
        obtaining a minimum-size dominating set extending our selection so far corresponds precisely to choosing for each $P = (v_1, \dots, v_p)$ in $\mathcal{P}$ whether we take a dominating set on it of size $\lceil p/3\rceil$ that includes $v_1$, or we take a dominating set of this size that includes $v_p$ (this is always possible by \Cref{obs:ds-path}) such that all remaining nodes in $V_3$ are dominated. 
        For the purpose of dominating vertices in $V_3$, the choice between including $v_1$ or $v_p$ in the dominating set for $P$ amounts to deciding between whether we use the dominating set on $P$ to dominate the neighbor of $v_1$ or the neighbor of $v_p$.
        In other words, for each $P\in\mathcal{P}$ we are deciding between a choice of the two edges incident to $P$ in $H$, and we want our final selection of edges to cover all nodes in $H$ corresponding to vertices in $V_3$ (if $P$ instead has unit-degree in $H$, then the choice of dominating set on it is forced, and if $P$ has degree zero in $H$, then we can take an arbitrary minimum-size dominating set of $P$).

        We can make this selection, and thus decide the remaining vertices to include in the dominating set, in polynomial time just by running a maximum matching algorithm on $H$.

        This completes the description of the dominating set algorithm.
        We spend polynomial time for each guess of $D_3$,
        and so the algorithm runs in $O^*(2^{n_3})$ time as claimed.
        The algorithm is correct because in the case where we guess $D_3 = D\cap V_3$ for some minimum-size dominating set $D$ that maximizes the value of $D\cap V_3$,
        the analysis above shows that we correctly extend $D_3$ to a minimum-size dominating set for all of $G'$.
    \end{claimproof}

    Having established \Cref{claim:forsure}, we are now ready to prove the desired runtime bound for finding a dominating set of minimum size in $G'$.

        Let $m'$ be the number of edges in $G'$.
    It is known that we can find a minimum-size dominating set in  $O^*(1.4969^n) \le O^*(2^{3n/5})$ time and polynomial space \cite{DBLP:journals/dam/RooijB11}.
    
    If $m' > 6n/5$, then this algorithm already finds a minimum dominating set in $G'$
    in the desired runtime bound. 

    Suppose  instead that we have $m' \le 6n/5$.
    
    By definition, $m' = m - n_1$.
    From  the handshaking lemma, we get that 
    \[2n_2 + 3n_{3} \le 2m' \le (m - n_1) + 6n/5.\]
    Since $n_1 + n_2 + n_{3} = n$, we have 
    \begin{align*}
        n_3 &= 2n_1 + (2n_2 + 3n_3) - 2(n_1 + n_2 + n_3) \\
        &\le \left(n_1 + n/5 \right) + \left(m - n_1 + 6n/5 \right) - 2n = m - 3n/5
    \end{align*}
    where we used the assumptions  that $n_1\le n/5$ and $m'\le 6n/5$. 
    
    Consequently, in this case the algorithm from \Cref{claim:forsure} finds a minimum size dominating set in   $O^*(2^{m - 3n/5})$ time and polynomial space, as desired. 
\end{proof}

\end{document}